\documentclass[11pt, english, a4paper]{article}

\usepackage{bbm}

\usepackage[utf8]{inputenc}
\usepackage[T1]{fontenc}
\usepackage[a4paper, lmargin=1.0in,rmargin=1.0in,bottom=1.0in,top=1.0in,twoside=False]{geometry}
\usepackage{fullpage}
\usepackage{libertine} % text font
\usepackage{inconsolata} % texttt font
\usepackage{amsmath, amsthm, amssymb}
\usepackage{graphicx}

\usepackage[shortlabels]{enumitem}
\usepackage{authblk}
\usepackage{thmtools}
\usepackage{thm-restate}
\usepackage{authblk}
\usepackage[hyphens]{xurl}
\usepackage{float}
\usepackage{xcolor}
\usepackage{mathtools}
\usepackage{setspace}

\usepackage{todonotes}

\usepackage{algorithm}% http://ctan.org/pkg/algorithm
\PassOptionsToPackage{noend}{algpseudocode}% comment/uncomment depending on whether want ends to show
\usepackage{algpseudocode}% http://ctan.org/pkg/algorithmicx
\usepackage[colorlinks=true, linkcolor = bordeaux, citecolor = darkblue, urlcolor = darkblue]{hyperref}

\makeatletter
\@addtoreset{ALG@line}{algorithm}
\makeatother

\usepackage{etoolbox}
\usepackage{tikz}
\usetikzlibrary{tikzmark}
\usetikzlibrary{calc}

\errorcontextlines\maxdimen

\newcommand{\ALGtikzmarkcolor}{black}% customise this, if you want
\newcommand{\ALGtikzmarkextraindent}{4pt}% customise this, if you want
\newcommand{\ALGtikzmarkverticaloffsetstart}{-.5ex}% customise this, if you want
\newcommand{\ALGtikzmarkverticaloffsetend}{-.5ex}% customise this, if you want
\makeatletter
\newcounter{ALG@tikzmark@tempcnta}

\newcommand\ALG@tikzmark@start{%
    \global\let\ALG@tikzmark@last\ALG@tikzmark@starttext%
    \expandafter\edef\csname ALG@tikzmark@\theALG@nested\endcsname{\theALG@tikzmark@tempcnta}%
    \tikzmark{ALG@tikzmark@start@\csname ALG@tikzmark@\theALG@nested\endcsname}%
    \addtocounter{ALG@tikzmark@tempcnta}{1}%
}

\def\ALG@tikzmark@starttext{start}
\newcommand\ALG@tikzmark@end{%
    \ifx\ALG@tikzmark@last\ALG@tikzmark@starttext
    \else
        \tikzmark{ALG@tikzmark@end@\csname ALG@tikzmark@\theALG@nested\endcsname}%
        \tikz[overlay,remember picture] \draw[\ALGtikzmarkcolor] let \p{S}=($(pic cs:ALG@tikzmark@start@\csname ALG@tikzmark@\theALG@nested\endcsname)+(\ALGtikzmarkextraindent,\ALGtikzmarkverticaloffsetstart)$), \p{E}=($(pic cs:ALG@tikzmark@end@\csname ALG@tikzmark@\theALG@nested\endcsname)+(\ALGtikzmarkextraindent,\ALGtikzmarkverticaloffsetend)$) in (\x{S},\y{S})--(\x{S},\y{E});%
    \fi
    \gdef\ALG@tikzmark@last{end}%
}

\apptocmd{\ALG@beginblock}{\ALG@tikzmark@start}{}{\errmessage{failed to patch}}
\pretocmd{\ALG@endblock}{\ALG@tikzmark@end}{}{\errmessage{failed to patch}}
\makeatother
\declaretheorem[name=Lemma, numberwithin = section]{lemma}
\declaretheorem[name=Theorem,sibling = lemma]{theorem}

\declaretheorem[name=Definition, sibling=lemma]{definition}
\declaretheorem[name=Corollary, sibling=lemma]{corollary}

\declaretheorem[name=Claim]{claim}

\usepackage[noabbrev,capitalise,nameinlink]{cleveref}
\crefname{claim}{Claim}{Claims}
\crefname{lemma}{Lemma}{Lemmas}
\crefname{theorem}{Theorem}{Theorems}
\crefname{proposition}{Proposition}{Propositions}
\crefname{question}{Question}{Questions}
\crefname{definition}{Definition}{Definitions}
\crefname{conjecture}{Conjecture}{Conjectures}
\crefname{observation}{Observation}{Observations}
\crefname{corollary}{Corollary}{Corollaries}
\crefformat{equation}{(#2#1#3)}
\Crefformat{equation}{Equation #2(#1)#3}
\crefname{remark}{Remark}{Remarks}
\crefname{scenario}{Scenario}{Scenarios}

\def\cqedsymbol{\ifmmode$\lrcorner$\else{\unskip\nobreak\hfil
\penalty50\hskip1em\null\nobreak\hfil$\lrcorner$
\parfillskip=0pt\finalhyphendemerits=0\endgraf}\fi}

\newenvironment{subproof}[1][\proofname]{%
  \begin{proof}[#1]%
}{%
  \end{proof}%
}

\setlist[itemize]{topsep=0ex,itemsep=0ex,parsep=0.25ex}
\setlist[enumerate]{topsep=0ex,itemsep=0ex,parsep=0.25ex}

\definecolor{bordeaux}{RGB}{100,0,50}
\definecolor{darkblue}{RGB}{25, 25, 112}

\newcommand{\mainR}{r}
\newcommand{\tinyR}{\rho}
\newcommand{\partition}{\mathcal{P}}
\newcommand{\WReach}[5]{\mathrm{WReach}_{#1}[#2,#3,#4,#5]}
\newcommand{\WReachS}[1]{\WReach{\mainR}{G}{\partition}{\preceq}{#1}}
\newcommand{\wcol}{\mathrm{wcol}}
\newcommand{\dom}{\mathrm{dom}}
\newcommand{\dist}{\mathrm{dist}}
\newcommand{\treewidth}{\mathrm{tw}}
\newcommand{\topm}{\mathrm{top}}

\title{Bounding $\varepsilon$-scatter dimension via metric sparsity%
\thanks{Marcin is supported by Polish National Science Centre SONATA BIS-12 grant number 2022/46/E/ST6/00143.}}

\author{Romain Bourneuf\thanks{\'Ecole Normale Sup\'erieure de Lyon, LIP, Lyon, France, \texttt{romain.bourneuf@ens-lyon.fr}}
\and 
Marcin Pilipczuk\thanks{Institute of Informatics, University of Warsaw, Poland, \texttt{m.pilipczuk@uw.edu.pl}}}

\date{}

\begin{document}
\maketitle
\begin{abstract}
    A recent work of Abbasi et al. [FOCS 2023] introduced the notion of \emph{$\varepsilon$-scatter dimension} of a metric space and
    showed a general framework for efficient parameterized approximation schemes (so-called EPASes) for a wide range of clustering problems
    in classes of metric spaces that admit a bound on the $\varepsilon$-scatter dimension.
    Our main result is such a bound 
    for metrics induced by graphs from any fixed proper minor-closed graph class.
    The bound is double-exponential in $\varepsilon^{-1}$ and the Hadwiger number
    of the graph class and is accompanied by a nearly tight lower bound that
    holds even in graph classes of bounded treewidth. 

    On the way to the main result, we introduce metric analogs of well-known graph
    invariants from the theory of \emph{sparsity}, including generalized coloring numbers
    and flatness (aka uniform quasi-wideness), and show bounds for these invariants
    in proper minor-closed graph classes.

    Finally, we show the power of newly introduced toolbox by showing a coreset
    for \textsc{$k$-Center} in any proper minor-closed graph class whose size is polynomial
    in $k$ (but the exponent of the polynomial depends on the graph class and $\varepsilon^{-1}$). 
\end{abstract}

\section{Introduction}
In this work, our focus is on the family of \emph{$k$-clustering} problems, where the
goal is to partition the input data points into $k$ clusters and choose a \emph{center} for
each cluster to minimize some objective function.
More precisely, we consider the \emph{Norm $k$-Clustering} problem, 
where the input
consists of a metric space $(V,\dist)$, a set of data points (called sometimes \emph{clients})
$P \subseteq V$, a set of potential centers (called sometimes \emph{facilities}) $F \subseteq V$, an integer $k$, 
and an objective function $f: \mathbb{R}^P \to \mathbb{R}$ that is a monotone norm,
and the goal is to minimize $f( (\dist(p, X))_{p \in P} )$ over all choices 
of centers $X \subseteq F$, $|X| \leq k$. Here, $\dist(p,X) = \min_{x \in X} \dist(p,x)$
is the distance from the point $p$ to its closest center. 

The \textsc{Norm $k$-Clustering} problem captures a wide range of clustering problems
appearing in many branches of computer science, from data mining through machine learning to computational geometry and optimization. 
For example, we get \textsc{$k$-Center} for $f( (a_p)_{p \in P}) = \max_{p \in P} a_p$,
\textsc{$k$-Median} for $f( (a_p)_{p \in P}) = \sum_{p \in P} a_p$, and
\textsc{$k$-Means} for $f( (a_p)_{p \in P}) = \sum_{p \in P} a_p^2$.

Unfortunately, most variants of \textsc{Norm $k$-Clustering} 
turn out to be rather difficult 
in important classes of metrics, such as high-dimensional Euclidean spaces~\cite{AwasthiCKS15,dasgupta-clustering}. 
However, in many applications the desired number of clusters, $k$, is small.
This makes the paradigm of \emph{parameterized complexity} applicable in the form of \emph{parameterized approximation}.
These considerations lead to the notion of
an \emph{Efficient Parameterized Approximation Scheme} (EPAS for short).
An EPAS takes on input an additional accuracy parameter $\varepsilon > 0$,
runs in time $f(k,\varepsilon) \mathrm{poly}(|V|)$ for a computable function $f$,
and returns a $(1+\varepsilon)$-approximate solution. 
EPASes for most basic \textsc{Norm $k$-Clustering} problems have been known 
since 2002 for Euclidean spaces~\cite{BadoiuHI02,KumarSS10}; results for other metrics include~\cite{BakerBHJK020,BravermanJKW21,KatsikarelisLP19}.

Abbasi et al.~\cite{scatter-dim} presented at FOCS 2023 a unifying concept explaining the existence
of EPASes for all \textsc{Norm $k$-Clustering} problems in many structured classes
of metric spaces. 
Their main conceptual contribution is the notion of \emph{$\varepsilon$-scatter dimension}
and a general framework that turns a bound on this invariant of a metric space into an
EPAS for \textsc{Norm $k$-Clustering}. That is, any class of metric spaces that admits
a (even purely existential, i.e., not algorithmic) bound on $\varepsilon$-scatter dimension
admits also an EPAS for any \textsc{Norm $k$-Clustering} problem. 

\begin{definition}
Let $(V,\dist)$ be a metric space and $\varepsilon > 0$ be an accuracy parameter.
A \emph{$\varepsilon$-ladder} is a sequence $(x_i,p_i)_{i=1}^\ell$ of pairs of points of~$V$
such that for some $r > 0$ we have: $\forall_{1 \leq j < i \leq \ell} \dist(p_j,x_i) \leq r$
yet $\forall_{1 \leq i \leq \ell} \dist(p_i,x_i) > (1+\varepsilon)r$.
The number of pairs $\ell$ is the \emph{length} of the $\varepsilon$-ladder
and the real $r$ is the \emph{width} of the $\varepsilon$-ladder. 
The \emph{$\varepsilon$-scatter dimension} of $(V,\dist)$ is the length of the longest
$\varepsilon$-ladder in $(V,\dist)$. 
\end{definition}

It is relatively easy to observe that $d$-dimensional Euclidean spaces
have $\varepsilon$-scatter dimension $(\Theta(\varepsilon^{-1}))^d$;
Abbasi et al.~\cite{scatter-dim} showed how to modify their framework to also capture high-dimensional
Euclidean spaces.

However, in many applications the metric space cannot be represented or well-approximated 
by a Euclidean space. A more versatile representation is a \emph{graph metric},
where $(V,\dist)$ is represented as an edge-weighted graph $G$ with $V(G) = V$
and $\dist(x,y)$ is the minimum distance between $x$ and $y$ in $G$. 
For example, distances in road networks can be represented as a graph metric with the underlying
graph being close to a planar graph. 
This motivates the study of properties of graph metrics induced by classes of sparse graphs,
such as planar graphs or, more generally, graphs with a fixed excluded minor. 

Abbasi et al.~\cite{scatter-dim} proved that graphs of bounded treewidth induce
metrics of bounded $\varepsilon$-scatter dimension via an elaborate argument.
Furthermore, they observed that combining the above with an advanced metric embedding result
of~\cite{Fox-EpsteinKS19}, one gets also a bound for graph metrics induced by planar graphs. 
Our main result is a clean and direct argument for any proper minor-closed graph class.
\begin{restatable}{theorem}{thmdim}\label{thm:dim-bound}
For every integer $h \geq 1$ and real $\varepsilon > 0$,
if $G$ is an edge-weighted $K_h$-minor-free graph, then the $\varepsilon$-scatter
dimension of the metric induced by $G$ is bounded by 
\[ \left(6c \cdot (9\varepsilon^{-1} + 2)^c\right)^{c+1}, \]
where 
\[ c = \binom{h-2 + 2 \cdot \lceil 36h\varepsilon^{-1} \rceil}{h-1} \cdot (12 + 72\varepsilon^{-1}) \cdot (h-1). \]
In particular, the $\varepsilon$-scatter dimension of $G$ is bounded
by $2^{(h\varepsilon^{-1})^{\mathcal{O}(h)}}$.
\end{restatable}
This in particular improves upon the triple-exponential (in $\varepsilon^{-1}$) bound for 
planar graphs and matches the double-exponential bound on treewidth of~\cite{scatter-dim}.
More importantly, it extends the scope of applicability of the framework of~\cite{scatter-dim}
to arbitrary proper minor-closed graph classes.
\begin{corollary}
    For every proper minor-closed graph class $\mathcal{G}$, 
    \textsc{Norm $k$-Clustering} admits an EPAS when restricted to metrics
    induced by graphs from $\mathcal{G}$.
\end{corollary}

We accompany \cref{thm:dim-bound} with an almost tight lower bound that shows
that the bound on the $\varepsilon$-scatter dimension cannot be improved even in classes
of bounded treewidth (recall that graphs of treewidth at most $h$ are $K_{h+2}$-minor-free). 
\begin{restatable}{theorem}{thmdimlb}\label{thm:dim-lb}
For every integers $t,r \geq 1$ and real $0 < \varepsilon < \frac{1}{r+2}$, 
there exists an edge-weighted graph of treewidth at most $2t+2$ that 
induces a metric of $\varepsilon$-scatter dimension at least $2^{\binom{t+r}{t}}$. 
\end{restatable}

Finally, we show that our techniques can also be used to develop new coresets in 
graph metrics. 
A \emph{coreset} is a (possibly weighted) small subset of clients that
approximates well the objective function for any set of at most $k$ centers. 
Coresets have been the leading technique for designing algorithms for clustering problems
for over a decade since the seminal works of Chen~\cite{Chen09} and Feldman and Langberg~\cite{FeldmanL11}.
While we know very good bounds on coresets for most problems in many classes
of metric spaces (see e.g.~\cite{BravermanCJKST022,Cohen-AddadLSS22}),
most of the toolbox is inapplicable to the special case of the \textsc{$k$-Center} problem.
Furthermore, the following example shows that one needs to consider
restricted classes of metrics to get interesting coreset results
for \textsc{$k$-Center}: let the metric consist of clients $p_1,\ldots,p_n$ and potential centers $x_1,\ldots,x_n$ with
$\dist(p_i,x_i) = 1+\varepsilon$ for $i \in [n]$ and
$\dist(p_i,x_j) = 1$ for $i,j \in [n]$, $i \neq j$. Then, 
for $k=1$, one needs to remember all clients 
in a coreset of any precision better than $1+\varepsilon$.%
\footnote{We thank Chris Schwiegelshohn for enlightening us with this example.}
Among restricted classes of metrics, we are only aware of a coreset for \textsc{$k$-Center} in doubling metrics~\cite{AghamolaeiG18}.
We fill in this gap by providing a coreset for \textsc{$k$-Center} in proper minor-closed
graph classes whose size is polynomial in the size bound $k$ (for fixed class and $\varepsilon > 0$).
\begin{theorem}\label{thm:coreset}
Given a \textsc{$k$-Center} instance, where the metric is given as a graph metric
on a $K_h$-minor-free graph, and an accuracy parameter $\varepsilon > 0$, one
can in polynomial time identify a subset $C$ of at most $k^{(h\varepsilon^{-1})^{\mathcal{O}(h)}}$
clients such that for every set $X$ of at most $k$ potential centers
it holds that
\[ \max_{p \in P} \dist(p, X) \leq (1+\varepsilon) \max_{p \in C} \dist(p, X).\]
\end{theorem}

\subsection{Our techniques: metric sparsity}
Our main technical contribution is to adapt a number of concepts and proof techniques
from \emph{sparsity} ---
the (abstract, mostly graph-theoretical) theory of sparse graph classes --- 
to the world of graph metrics. 

Sparsity is the theory of sparse graph classes developed over the last nearly 20 years;
see the textbook~\cite{sparsity-book} or the more modern lecture notes~\cite{sparsitynotes}. 
The main contribution of this theory is the identification of \emph{bounded expansion}
and \emph{nowhere dense} graph classes as wide generalizations of proper minor-closed graph
classes that properly capture the notion of ``being sparse'' and allow for efficient algorithms.
In particular, the notion of nowhere dense correctly captures the border of tractability
of model checking first-order formulae in subgraph-closed graph classes~\cite{GroheKS17}.

The main focus of sparsity is on unweighted simple graphs and interactions that
happen within constant distance --- which corresponds to the expressive power of
first-order formulae. From the family of clustering problems, directly it is only 
able to speak about partitioning into a constant number of clusters with constant radius.
For example, the existence of a solution to the \textsc{$k$-Center} problem in an unweighted graph,
with maximum distance $r$ from any point to its closest center can be expressed as a first-order formula
whose length depends only on $k$ and $r$.
We can infer from~\cite{GroheKS17} that this problem can be solved in $f(k,r) \cdot n^{1+o(1)}$ time in nowhere dense graph classes, where $n$ is the size of the input graph
and $f$ is some computable function. However, the techniques seem to break down severely 
if one considers edge weights and/or superconstant distances. 

We revisit this state of the matter and propose a different angle to adapt the 
rich graph-theoretical and algorithmic toolbox of sparsity to graph metrics. 
Our inspiration comes from the following observation.
The main algorithm of Abbasi et al.~\cite{scatter-dim} is built upon a two-decade
old algorithm of Badoui, Har-Peled, and Indyk~\cite{BadoiuHI02}.
Essentially the same algorithm, but in the context of sparsity, appeared in 2019
in~\cite{FabianskiPST19}. In~\cite{FabianskiPST19},
the analogs of $\varepsilon$-ladders are called \emph{semi-ladders} and,
from their perspective, this definition comes from stability theory. 
Along similar lines as~\cite{scatter-dim}, they observe that a bound on the length 
of semi-ladders implies tractability of \textsc{$k$-Center}.

This led us to revisit the sparsity proof pipeline that proves bounds on the length
of semi-ladders in sparse graph classes and attempt to cast them onto the world
of graph metrics. This quest has been completed successfully: we nontrivially
adapted the notion
of \emph{weak coloring numbers} and \emph{flatness} (aka \emph{uniform quasi-wideness})
and followed the established routes for (tightly)
bounding these invariants in minor-closed graph classes.

\paragraph{Weak coloring numbers.}
One of the most useful notions in the theory of graph classes of bounded expansion is
so-called \emph{generalized coloring numbers}, most notably \emph{weak coloring numbers},
whose aim is to generalize degeneracy orderings to larger (but still constant) distances.

Let $G$ be a (simple, unweighted) graph, $r \geq 0$ be an integer, and let
$\preceq$ be a total ordering on $V(G)$.
We say that $u \in V(G)$ is \emph{weakly $\mainR$-reachable} from $v \in V(G)$ if there exists
a path $P$ from $v$ to $u$ of length at most $\mainR$ in $G[\{w \in V(G)~|~u \preceq w\}]$. 
Graphs from sparse graph classes, such as proper minor-closed graph classes, admit
orderings $\preceq$ where every vertex only has a bounded (as a function of $\mainR$)
number of weakly $\mainR$-reachable vertices, and such orderings are very useful algorithmically. 

The crucial insight into adapting this definition (and other definitions)
to metric spaces is that we should order not individual vertices, but subgraphs of small diameter. This leads to the following definition. 
\begin{definition}[weakly reachable sets in graph metric spaces]
Let $G$ be an edge-weighted graph, $\mainR > 0$ be a real, $\partition$
be a partition of~$V(G)$, and $\preceq$ be a total order on $\partition$. 
We say that $Y \in \partition$ is \emph{weakly $\mainR$-reachable} from $X \in \partition$
if there exists a path from a vertex of $X$ to a vertex of $Y$ of length at most $\mainR$
in the graph $G[\bigcup \{Z \in \partition~|~Y \preceq Z\}]$. 
By $\WReachS{X}$ we denote the set of those $Y \in \partition$ that are weakly $\mainR$-reachable from $X$. 

The \emph{weak coloring number} of $G$, $\partition$, and $\preceq$ is defined as 
\[ \wcol_{\mainR}(G,\partition,\preceq) := \max_{X \in \partition} |\WReachS{X}|. \]
\end{definition}

In the context of sparsity, provably tight bounds for weak coloring numbers in proper
minor-closed graph classes were provided in~\cite{HeuvelMQRS17} using so-called \emph{cop decompositions}~\cite{Andreae86}. 
A recent work of Chang et al.~\cite{buffered-cop-dec} introduced a variant called
\emph{buffered cop decompositions} that are meant for graph metrics. 
Building upon the ideas of~\cite{HeuvelMQRS17}, but using the buffered version of cop decompositions of~\cite{buffered-cop-dec}, in \cref{sec:wcol} we prove the following.
\begin{restatable}{theorem}{thmwcol}\label{thm:wcol}
Let $h > 0$ be an integer.
Given a $K_h$-minor-free edge-weighted graph $G$ and a real $\tinyR > 0$,
one can in polynomial time compute a partition $\mathcal{P}$ of $V(G)$ into parts
of (strong) diameter at most $\tinyR$ and a total order $\preceq$ on $\mathcal{P}$ such
that for every $\mainR > \tinyR$, it holds that
\[ \wcol_{\mainR}(G,\partition, \preceq) \leq \binom{h-2 + 2 \cdot \left\lceil 4h \frac{\mainR}{\tinyR} \right\rceil} {h-1} \cdot \left(12+8 \frac{\mainR}{\tinyR}\right) \cdot (h-1) =: c(h, \mainR/\tinyR). \]
\end{restatable}
One should think of $\mainR$ as being of the order of the relevant (e.g., the objective function
in the case of \textsc{$k$-Center}) distance from a client to a center,
while we would like every element of $\partition$ to be of small diameter: 
$\tinyR$ will be chosen smaller than the allowed approximation error of the designed algorithm. 
Thus, in our applications, the ratio $\frac{\mainR}{\tinyR}$ will be bounded by a function of the chosen approximation parameter $\varepsilon > 0$.

From the theory of sparsity, we know that there is a short step from weak coloring numbers to \emph{sparse covers}. For an edge-weighted graph $G$ and $\mainR > 0$, a \emph{$\mainR$-cover}
is a family $\mathcal{D}$ of subsets of $V(G)$ such that for every $v \in V(G)$ there exists
$D \in \mathcal{D}$ that contains all vertices within distance $\mainR$ from $v$. 
An $\mainR$-cover $\mathcal{D}$ has two quality parameters: 
the \emph{diameter blowup}, defined as $\frac{1}{\mainR} \max_{D \in \mathcal{D}} \mathrm{diameter}(G[D])$, and the \emph{overlap} or \emph{ply}, defined as
$\max_{v \in V(G)} |\{D \in \mathcal{D}~|~v \in D\}|$. 
\begin{lemma}\label{lem:wcol2cover}
    Let $G$ be an edge-weighted graph, $\mainR > \tinyR > 0$, 
    $\partition$ be a partition of $V(G)$ into pieces of weak diameter at most $\tinyR$, 
    and $\preceq$ be a total ordering of $\partition$.
    For $X \in \partition$, let
    \[ D(X) := \bigcup \{ Y \in \partition~|~X \in \WReach{2\mainR}{G}{\partition}{\preceq}{Y}\}\]
    and 
    \[ \mathcal{D} = \left\{D(X),~X \in \partition\right\}.\]
    Then, $\mathcal{D}$ is an $\mainR$-cover with diameter blowup at most $4+3\frac{\tinyR}{\mainR}$
    and overlap $\wcol_{2\mainR}(G,\partition,\preceq)$.
\end{lemma}
\begin{proof}
    To see that $\mathcal{D}$ is an $\mainR$-cover, pick $v \in V(G)$ and let 
    $X$ be the $\preceq$-minimum element of $\partition$ that contains a vertex $x$ within distance
    $\mainR$ from $v$. Then, for every other vertex $y$ within distance $\mainR$ from $v$,
    if $Y \in \partition$ is the part containing $y$, then $X \in \WReach{2\mainR}{G}{\partition}{\preceq}{Y}$, as witnessed by the concatenation of the shortest paths from $v$ to $x$ and $y$. 

    Since the (weak) diameter of each part of $\partition$ is bounded by
    $\tinyR$ and each part $Y \in \partition$ in $D(X)$ is within distance at most $2\mainR$
    from $X$, the diameter of each $D(X)$ is bounded by $2\cdot 2\mainR + 3\tinyR$.
    This gives the diameter blowup bound. The overlap bound is immediate, as each part $Y \in \partition$ is contained in $D(X)$ if and only if $X \in \WReach{2\mainR}{G}{\partition}{\preceq}{Y}$.
\end{proof}
Pipelining \cref{thm:wcol} with \cref{lem:wcol2cover} reproves
the result of Filtser~\cite{sparse-covers} that $K_h$-minor-free graphs admit
$\mainR$-covers of blowup $4+\varepsilon$ and overlap $\mathcal{O}(1/\varepsilon)^h$
for every $\varepsilon > 0$. 
(In fact, the proof of Filtser~\cite{sparse-covers} follows very similar lines
to our proofs, but without the identification of weak coloring numbers as a
clear intermediate step.)

Filtser~\cite{sparse-covers} used covers to develop low-diameter decompositions.
One of the prominent applications of low-diameter decompositions is approximation
algorithms for cut problems, such as \textsc{Multicut}. 
Recently, Friedrich et al.~\cite{treewidth} showed a $\mathcal{O}(\log \treewidth)$-approximation
for \textsc{Multicut}. 
Based on their arguments, Filtser et al.~\cite{DBLP:journals/corr/abs-2407-12230}
showed $\mainR$-covers of constant blowup and overlap $\mathrm{poly}(\mathrm{treewidth})$.
Their techniques inspired us 
to show a slightly improved bound for graphs of bounded treewidth
(but with a partition into sets of only bounded weak diameter),
presented in \cref{sec:wcol-tw}.
\begin{restatable}{theorem}{thmwcoltw}\label{thm:wcol-tw}
Given an edge-weighted graph $G$, a real $\tinyR > 0$, and a tree decomposition
of $G$ of maximum bag size at most $k$, one can in polynomial time compute
a partition $\partition$ of $V(G)$ into sets of weak diameter at most $\tinyR$
and a total order $\preceq$ on~$\partition$ such that for every $\mainR > \tinyR$
we have
\[ \wcol_{\mainR}(G,\partition,\preceq) \leq 
\min \left( k \cdot 2^k \cdot \binom{k + \left\lceil 2\frac{\mainR}{\tinyR} \right\rceil}{k},
\left(2 \lceil 2\mainR / \tinyR \rceil + k + 1\right)^{3\lceil 2\mainR/\tinyR \rceil + 4}\right)
. \]
\end{restatable}
We remark that the second bound of \cref{thm:wcol-tw}, together with \cref{lem:wcol2cover},
reproves the aforementioned result of~\cite{DBLP:journals/corr/abs-2407-12230} that graphs
of bounded treewidth admit $\mainR$-covers of constant blowup and polynomial-in-treewidth
overlap.

\paragraph{Flatness.}
A second very prolific notion from the theory of sparsity is \emph{flatness}, also called \emph{uniform quasi-wideness}. 
A set $A \subseteq V(G)$ is \emph{$\mainR$-scattered} if, for every distinct $u,v \in A$
the distance from $u$ to $v$ is larger than $\mainR$. Note that this definition makes sense
both in unweighted simple graphs (where $\mainR$ is usually a small constant) and in
metric spaces (where $\mainR$ is a positive real). 

The notion of flatness is motivated by the following observation.
Fix a small integer $\mainR > 0$. In sparse graphs, large
sets of vertices contain big sets that are almost $\mainR$-scattered: they become $\mainR$-scattered after the deletion of a few ``hub-like'' vertices from the graph. 
Consider for example a star (a hub vertex with many degree-1 neighbors): no two leaves of the star are $2$-scattered, but after deleting the center of the star, one obtains a large $2$-scattered set consisting of all leaves. Flatness captures the intuition that such a star picture is essentially the only bad picture that can happen in sparse graphs.

Formally, we say that a graph class $\mathcal{G}$ is \emph{flat} if for every integer $\mainR > 0$ there exists $s =s(\mainR)$ such that for every integer $m \geq 0$ there exists $M = M(m,\mainR)$ such that the following holds. For every $G \in \mathcal{G}$ and $A \subseteq V(G)$ of size at least $M$, there exists $S \subseteq V(G)$ of size at most $s$ and $B \subseteq A \setminus S$ of size at least $m$ that is $\mainR$-scattered in $G-S$. 
The very important part of this definition is that $s$, the number of ``hub-like'' vertices to be deleted,
depends only on $\mainR$, but not on the desired number of vertices $m$. 

Nadara et al.~\cite{empirical-uqw} showed a short argument turning orderings
of small weak coloring number into flatness. 
Using the paradigm that individual vertices should become small-diameter sets,
in \cref{sec:flatness} we obtain the following metric analog.
\begin{restatable}{theorem}{thmflat}\label{thm:wcol2flat}
    Given an edge-weighted graph $G$, a real $\mainR > 0$, a partition $\mathcal{P}$ of $V(G)$,
    a total ordering $\preceq$ of $\partition$, an integer $m \geq 0$,
    and a set $A \subseteq \partition$ such that $|A| \geq (2mc)^{c+1}$ 
    where $c = \wcol_{\mainR}(G,\partition,\preceq)$, 
    one can in polynomial time find $S \subseteq \partition$ and $B \subseteq A \setminus S$
    such that $|S| \leq c$, $|B| \geq m$, and every two distinct $X_1,X_2 \in B$
    are at distance more than $\mainR$ in $G-\bigcup S$.
\end{restatable}

Combining it with \cref{thm:wcol}, and observing that
every element of $\partition$ given by \cref{thm:wcol}
can contain at most one vertex of a $\tinyR$-scattered set, 
we immediately obtain the following.%
\footnote{Recall that for a subset $A \subseteq V(G)$, the \emph{strong diameter}
of $A$ is defined as $\max_{u,v \in A} \dist_{G[A]}(u,v)$, while the \emph{weak diameter}
of $A$ is defined as $\max_{u,v \in A} \dist_G(u,v)$. That is, in the strong diameter definition
we measure the distance in $G[A]$, while in the weak diameter definition we measure the distance
in the whole $G$.}
\begin{restatable}[flatness of proper minor-closed graph metrics]{corollary}{corflat}\label{cor:flat}
  Given an edge-weighted $K_h$-minor-free graph $G$, two reals $\mainR > \tinyR > 0$, 
  an integer $m \geq 0$, and a $\tinyR$-scattered set $A \subseteq V(G)$ of size
  at least $(2mc)^{c+1}$, where $c = c(h,\mainR/\tinyR)$ is defined in \cref{thm:wcol}, 
  one can in polynomial time compute a set $S$ consisting of at most $c$ subsets of $V(G)$
  of strong diameter at most $\tinyR$ each and a set $B \subseteq A \setminus \bigcup S$ 
  of size at least $m$ that is $\mainR$-scattered in $G-\bigcup S$. 
\end{restatable}

\paragraph{Bounding $\varepsilon$-scatter dimension.}
Armed with the flatness statement of \cref{cor:flat}, the proof of
\cref{thm:dim-bound} is relatively simple and presented in \cref{sec:ladders}. 

Grohe et al.~\cite{GroheKRSS18} presented a construction that shows that known 
(sparsity) bounds for weak coloring numbers in minor-closed graph classes are essentially optimal
and Nadara et al.~\cite{empirical-uqw} observed that the same construction also gives lower bounds
for flatness guarantees. 
In \cref{sec:lb} we modify this construction to prove also the lower bound 
of \cref{thm:dim-lb} for $\varepsilon$-scatter dimension. 

\paragraph{Further consequences.}
Our main technical contribution lies in the introduction of the metric analogs
of the main concepts from sparsity and showing that they can be used to obtain
an (almost tight, thanks to \cref{thm:dim-lb}) bound for $\varepsilon$-scatter
dimension in proper minor-closed graph classes. 

The natural next question is about further applications of the introduced toolbox. 
In \cref{sec:kcenter}, we show improvements for the special case
of \textsc{$k$-Center} and prove \cref{thm:coreset}. 
In short, the proof uses the developed notion of metric flatness to control 
the number of clients ``useful'' for an algorithm as in~\cite{BadoiuHI02}. 

\paragraph{Notation.}
We use the following notation from set theory: for a family $S$ of sets, $\bigcup S$
is a shorthand for $\bigcup_{A \in S} A$.

For a metric space $(V,\dist)$, and sets $X,Y \subseteq V$, we set
$\dist(X,Y) \coloneqq \min_{x \in X} \min_{y \in Y} \dist(x,y)$, and similarly $\dist(x,Y) \coloneqq \min_{y \in Y} \dist(x,y)$
for $x \in V$ and $Y \subseteq V$.

Let $G$ be a graph. Two sets $A,B \subseteq V(G)$ are \emph{adjacent} if there exists an edge with one endpoint in $A$
and the second endpoint in $B$. 
A \emph{tree decomposition} of $G$ is a pair $(T,\beta)$ where $T$ is a tree
and $\beta$ assigns to every $t \in V(T)$ a set $\beta(t) \subseteq V(G)$ called a \emph{bag}
such that (1) for every $v \in V(G)$, $\{t\in V(T)~|~v \in \beta(t)\}$ induces a nonempty connected
subgraph of $T$, and (2) for every $uv \in E(G)$, there exists $t  \in V(T)$ with $u,v \in \beta(t)$. The \emph{width} of $(T,\beta)$ is $\max_{t \in V(T)} |\beta(t)|-1$ and
the \emph{treewidth} of $G$ is the minimum width of a tree decomposition of $G$.

%% new organization
\section{Weak coloring numbers}\label{sec:wcol}
As discussed, to prove \cref{thm:wcol}, we need the buffered cop decompositions
of~\cite{buffered-cop-dec}. In the following definitions from~\cite{buffered-cop-dec},
$G$ is an edge-weighted graph. 
\begin{definition}
    A \emph{supernode} $\eta$ is a pair $(V_{\eta}, T_{\eta})$ where $V_{\eta} \subseteq V(G)$ and $T_{\eta}$ is a tree in $G[V_{\eta}]$ called the \emph{skeleton} of $\eta$. 
    $\eta$ has \emph{radius} $\Delta$ if every $v \in V_{\eta}$ is at distance at most $\Delta$ from $T_{\eta}$ in $G[V_{\eta}]$. 
\end{definition}

\begin{definition}
    A \emph{buffered cop decomposition} for $G$ is a pair $(\mathcal{Q}, \mathcal{T})$ where $\mathcal{Q}$ is a family of supernodes such that $\{V_{\eta}, \eta \in \mathcal{Q}\}$ is a partition of  $V(G)$ and $\mathcal{T}$ is a rooted tree, called the \emph{partition tree}, whose nodes are the supernodes in $\mathcal{Q}$. 
    For every supernode $\eta \in \mathcal{Q}$, we denote by $\dom(\eta)$ the set of all vertices of $G$ which belong to a set $V_{\zeta}$ such that $\zeta \in \mathcal{Q}$ is a descendant of $\eta$ in $\mathcal{T}$.
\end{definition}

\begin{definition}
    A \emph{$(\Delta, \gamma, w)$-buffered cop decomposition} for $G$ is a buffered cop decomposition $(\mathcal{Q}, \mathcal{T})$ satisfying the following properties: \begin{description}[itemsep=0px,topsep=0px]
        \item[Supernode radius.] Every supernode $\eta \in \mathcal{Q}$ has radius $\Delta$.
        \item[Shortest path skeleton.] For every supernode $\eta \in \mathcal{Q}$, the skeleton $T_{\eta}$ is a single-source shortest path tree in $G[\dom(\eta)]$ with at most $w$ leaves (not counting the root).
        \item[Supernode buffer.] Let $\eta \in \mathcal{Q}$ be a supernode and let $\zeta \in \mathcal{Q}$ be a supernode that is an ancestor of $\eta$ in $\mathcal{T}$. Either $V_{\eta}$ and $V_{\zeta}$ are adjacent in $G$ or $\dist_{G[\dom(\zeta)]}(V_{\eta}, V_{\zeta}) > \gamma$.
        \item[Tree decomposition.] For every $\eta \in \mathcal{Q}$, the set $\mathcal{A}_{\eta}$ of all supernodes $\zeta \in \mathcal{Q}$ such that $\zeta$ is an ancestor of $\eta$ in $\mathcal{T}$ and $V_{\zeta}$ and $V_{\eta}$ are adjacent in $G$ has size at most $w$ (including $\eta$). Furthermore, setting $W_{\eta} = \bigcup_{\zeta \in \mathcal{A}_{\eta}} V_{\zeta}$, we have that $(\mathcal{T}, (W_{\eta})_{\eta \in \mathcal{Q}})$ is a tree decomposition of $G$.
    \end{description}
\end{definition}

\begin{theorem}[\cite{buffered-cop-dec}] \label{th:buffered-cop-dec}
    Given a $K_h$-minor-free edge-weighted graph $G$ and a real $\mainR > 0$, 
    one can in polynomial time
    compute a $(\mainR, \mainR/h, h-1)$-buffered cop decomposition of $G$.
\end{theorem}

Buffered cop decompositions were recently used by Filtser~\cite{sparse-covers}
to show sparse covers of graphs from proper minor-closed graph classes. 
We will need the following lemma of~\cite{sparse-covers}, whose proof is 
in fact very similar to analogous proofs from the sparsity world, cf.~\cite{GroheKRSS18,HeuvelMQRS17}.

\begin{lemma}[Lemmata~1 and~2 of~\cite{sparse-covers}]\label{lem:few-supernodes-close}
    Let $G$ be an edge-weighted graph,
    $(\mathcal{Q}, \mathcal{T})$ be a $(\Delta, \gamma, w)$-buffered cop decomposition of $G$,
    and let $\eta \in \mathcal{Q}$ be a supernode. 
    Then, for every integer $q \geq 1$ there are at most $\binom{w + 2q - 1}{w}$ ancestors $\zeta$ of $\eta$ in $\mathcal{T}$ that satisfy
    \[ \dist_{G[\dom(\zeta)]}(V_{\eta}, V_{\zeta}) \leq q \cdot \gamma. \]
\end{lemma}

Armed with the buffered cop decompositions, we are ready to prove \cref{thm:wcol}.

\thmwcol*

\begin{proof}
    Set $\Delta = \tinyR/4$.
    Since $G$ is $K_h$-minor-free, in polynomial time 
    we can compute a $(\Delta, \Delta/h, h~-~1)$-buffered cop decomposition $(\mathcal{Q}, \mathcal{T})$ of $G$ by \cref{th:buffered-cop-dec}.
    For every supernode $\eta \in \mathcal{Q}$, construct greedily a maximal set $S_{\eta} \subseteq V(T_{\eta})$ that is $\Delta$-scattered in $G[V_{\eta}]$; we have that for every vertex $u \in V(T_{\eta})$, there exists a vertex $v \in S_{\eta}$ such that $\dist_{G[V_{\eta}]}(u, v) \leq \Delta$.
    We build a partition $\mathcal{P}_{\eta} = \{X_v, v \in S_{\eta}\}$ of $V_{\eta}$ as follows: for every $u \in V_{\eta}$, let $v \in S_{\eta}$ be such that $\dist_{G[V_{\eta}]}(u, v)$ is minimum (breaking ties consistently, i.e., according to some fixed total order on $S_\eta$) and add $u$ to $X_v$.
    Finally, we set $\partition = \bigcup_{\eta \in \mathcal{Q}} \mathcal{P}_{\eta}$.
    
    Recall that $\mathcal{T}$ is a rooted tree whose nodes are the supernodes. Fix any order $\preceq$ on $\partition$ such that if $\eta$ is a proper ancestor of $\zeta$ in $\mathcal{T}$ and $X \in \mathcal{P}_{\eta}, Y \in \mathcal{P}_{\zeta}$ then $X \prec Y$. Observe that $\partition$ and $\preceq$ can indeed be computed in time $|V(G)|^{O(1)}$. We now show that they satisfy the desired properties.

    $(\mathcal{Q}, \mathcal{T})$ is a buffered cop decomposition so $\{V_{\eta}, \eta \in \mathcal{Q}\}$ is a partition of $V(G)$. Since every $\mathcal{P}_{\eta}$ is a partition of $V_{\eta}$ then $\partition$ is a partition of $V(G)$.
    Let $X \in \partition$. There exist $\eta \in \mathcal{Q}$ and $v \in S_{\eta}$ such that $X = X_v$. 
    Let $u \in X_v$. Then, $u \in V_{\eta}$ so by the supernode radius property there exists $w \in V(T_{\eta})$ such that $\dist_{G[V_{\eta}]}(u, w) \leq \Delta$.
    By construction of $S_{\eta}$, there exists $v' \in S_{\eta}$ such that $\dist_{G[V_{\eta}]}(w, v') \leq \Delta$.
    Thus, $v' \in S_{\eta}$ and $\dist_{G[V_{\eta}]}(u, v') \leq 2\Delta$,
    so since $u$ was added to $X_v$ then $\dist_{G[V_{\eta}]}(u, v) \leq \dist_{G[V_{\eta}]}(u, v') \leq 2\Delta$.
    Furthermore, every $u'$ in the shortest path from $u$ to $v$ in $G[V_{\eta}]$ was also added to $X_v$ since we broke ties consistently. Therefore, there is a path of length at most $2\Delta$ in $G[X_v]$ from $u$ to $v$. This proves that $G[X]$ has 
    diameter at most $4\Delta = \tinyR$.
    
    Fix some $\mainR > \tinyR$; we now prove the promised bound
    on $\wcol_{\mainR}(G,\partition,\preceq)$. 
    Fix $X \in \partition$ and let us inspect $\WReachS{X}$. 
    
    Consider $Y \in \WReachS{X}$. We have $Y \in \partition$, $Y \preceq X$ and there is a path $P$ of length at most $\mainR$ from $X$ to $Y$ in $G - \left( \bigcup_{Z \prec Y}Z\right)$.
    Let $\eta_X \in \mathcal{Q}$ (resp. $\eta_Y \in \mathcal{Q}$) be such that $X \in \mathcal{P}_{\eta_X}$ (resp. $Y \in \mathcal{P}_{\eta_Y}$).
    We first argue that $\eta_Y$ and $\eta_X$ are comparable in $\mathcal{T}$.
    By contradiction, suppose not and let $\eta$ be their lowest common ancestor in $\mathcal{T}$.
    We consider the tree decomposition $(\mathcal{T}, (W_{\zeta})_{\zeta \in \mathcal{Q}})$ given by the tree decomposition property of the buffered cop decomposition.
    $P$ is a path from $X \subseteq V_{\eta_X} \subseteq W_{\eta_X}$ to $Y \subseteq V_{\eta_Y} \subseteq W_{\eta_Y}$ and $\eta$ separates $\eta_X$ and $\eta_Y$ in $\mathcal{T}$ so there exists $v \in V(P) \cap W_{\eta}$. 
    There exists an ancestor $\zeta$ of $\eta$ such that $v \in V_{\zeta}$, and there exists $Z \in \mathcal{P}_{\zeta}$ such that $v \in Z$.
    Since $\zeta$ is an ancestor of $\eta$, $\zeta$ is a proper ancestor of $\eta_Y$ so $Z \prec Y$ by definition of $\preceq$.
    This contradicts the fact that $V(P) \subseteq V(G) - \left( \bigcup_{Z \prec Y}Z\right)$.
    
    Thus, $\eta_Y$ and $\eta_X$ are comparable in $\mathcal{T}$, and $Y \preceq X$ implies that $\eta_Y$ is an ancestor of $\eta_X$ in $\mathcal{T}$.
    We now prove that $V(P) \subseteq \dom(\eta_Y)$.
    By contradiction, suppose not and consider the vertex $v \in V(P) \setminus \dom(\eta_Y)$ closest to $X$. Let $\eta_v \in \mathcal{Q}$ be such that $v \in V_{\eta_v}$. 
    Since $v \notin \dom(\eta_Y)$ then $\eta_v$ is not a descendant of $\eta_Y$. However, $\eta_X$ is a descendant of $\eta_Y$ so $v \notin V_{\eta_X}$, hence $v \notin X$ and in particular $v$ is not the first vertex of $P$. 
    Let $u \in \dom(\eta_Y)$ be the predecessor of $v$ in $P$, and $\eta_u \in \mathcal{Q}$ be such that $u \in V_{\eta_u}$. Since $u \in \dom(\eta_Y)$, it follows that $\eta_u$ is a descendant of $\eta_Y$.
    By the tree decomposition property, there exists a supernode $\zeta$ such that $u, v \in W_{\zeta}$, hence $\eta_u$ and $\eta_v$ are ancestors of $\zeta$ in $\mathcal{T}$. Since $\eta_u$ is a descendant of $\eta_Y$ in $\mathcal{T}$ and not $\eta_v$ then $\eta_v$ is an ancestor of $\eta_u$. Therefore, $\eta_v$ and $\eta_Y$ are comparable in $\mathcal{T}$ and $\eta_v$ is not a descendant of $\eta_Y$ so $\eta_v$ is a proper ancestor of $\eta_Y$.
    Let $Z \in \mathcal{P}_{\eta_v}$ be such that $v \in Z$. Then, $Z \prec Y$ by definition of $\preceq$.
    This contradicts the fact that $V(P) \subseteq V(G) - \left( \bigcup_{Z \prec Y}Z\right)$.
    Therefore, $V(P) \subseteq \dom(\eta_Y)$ and thus $\dist_{G[\dom(\eta_Y)]}(X, Y) \leq \mainR$ so $\dist_{G[\dom(\eta_Y)]}(V_{\eta_X}, V_{\eta_Y}) \leq \mainR$.
    
    By \cref{lem:few-supernodes-close} for 
    $q = \lceil \frac{\mainR}{\Delta/h} \rceil = \lceil 4h\frac{\mainR}{\tinyR} \rceil$,
    the number of such supernodes $\eta_Y$ is bounded by
    \[ \binom{h-2+2 \cdot \left\lceil 4h \frac{\mainR}{\tinyR} \right\rceil}{h-1}. \]
    To finish the proof, it suffices to show that for every supernode $\eta$ it holds that
    \begin{equation}\label{eq:single-eta-bound} \left| \WReachS{X} \cap \partition_{\eta} \right| \leq \left(12+8\frac{\mainR}{\tinyR}\right) \cdot (h-1). 
    \end{equation}
    Assume the contrary; let $\eta$ violate~\eqref{eq:single-eta-bound}.
    For each $Y \in \mathcal{P}_{\eta}$, let $c_Y \in S_{\eta} \cap Y$ be the center
    vertex of $Y$; recall that $c_Y$ lies on $T_{\eta}$. 
    By the shortest path skeleton property, $T_{\eta}$ is a single-source shortest path tree
    in $G[\dom(\eta)]$ with at most $h-1$ leaves, so one of the root-to-leaf paths in $T_{\eta}$ contains more than $12 + 8\frac{\mainR}{\tinyR}$ vertices $c_Y$ for
    $Y \in \WReachS{X} \cap \partition_{\eta}$. 
    Write them as $c_0, \ldots, c_k$ according to the order in which they appear along this path and let $X_0, \ldots, X_k$ be the corresponding parts of $\mathcal{P}$
    (i.e., $c_i$ is a shorthand for $c_{X_i}$ and also $k \geq 12+8\frac{\mainR}{\tinyR}$).
    For $i \neq j$, we have $\dist_{G[\dom(\eta)]}(c_i, c_j) = \dist_{G[V(T_{\eta})]}(c_i, c_j)$ since $T_{\eta}$ is a single-source shortest path tree in $G[\dom(\eta)]$, and $c_i$ and $c_j$ are in the same root-to-leaf path. Therefore, $\dist_{G[\dom(\eta)]}(c_i, c_j) = \dist_{G[V_{\eta}]}(c_i, c_j) > \Delta$ by the construction of $S_{\eta}$.
    Since the $c_i$ are pairwise at distance greater than $\Delta$ in $G[\dom(\eta)]$ and are along a shortest path in $G[\dom(\eta)]$, an immediate induction shows that \[ \dist_{G[\dom(\eta)]}(c_0, c_k) > k \cdot \Delta \geq 2\mainR + 3\tinyR. \]
    However, 
    \[ \dist_{G[\dom(\eta)]}(c_0, c_k) \leq \tinyR + \dist_{G[\dom(\eta)]}(X_0, X) + \tinyR + \dist_{G[\dom(\eta)]}(X, X_k) + \tinyR \leq 2\mainR + 3\tinyR. \]
    This proves~\eqref{eq:single-eta-bound} and concludes the proof of \cref{thm:wcol}.
\end{proof}

\section{Flatness}\label{sec:flatness}
Having proven \cref{thm:wcol}, we are now ready to prove \cref{thm:wcol2flat}.
Its proof closely follows the lines of an analogous proof of~\cite{empirical-uqw}.

\thmflat*

\begin{proof}
    The algorithm is presented as \cref{alg:uqw}. 

    \begin{algorithm}
    \caption{Algorithm for \cref{thm:wcol2flat}}
    \label{alg:uqw}
    \begin{algorithmic}[1]
    \Require Edge-weighted graph $G$, $\mainR > 0$, 
    partition $\partition$ of $V(G)$,  ordering $\preceq$ of $\partition$,
    and $A \subseteq \partition$

    \State $S \gets \emptyset$
    
    \While{$\exists_{X \in \mathcal{P} \setminus S}
     |\{Y \in A_, X \in \WReach{\mainR}{G - \bigcup S}{\partition \setminus S}{\preceq}{Y}\}| \geq \frac{|A|}{2mc}$}
        \State $A \gets \{Y \in A, 
        X \in \WReach{\mainR}{G - \bigcup S}{\partition \setminus S}{\preceq}{Y}\}$
        \State $S \gets S \cup \{X\}$
    \EndWhile

    \State $B \gets \emptyset$
    \State $A' \gets A \setminus S$

    \While{$A' \neq \emptyset$}
        \State $X \gets$ arbitrary element of $A'$
        \State $B \gets B \cup \{X\}$
        \State $A' \gets \{Y \in A',
        \WReach{\mainR}{G - \bigcup S}{\partition\setminus S}{\preceq}{Y} \cap \WReach{\mainR}{G - \bigcup S}{\partition\setminus S}{\preceq}{X} = \emptyset\}$
    \EndWhile

    \State Return $S, B$

    \end{algorithmic}
    \end{algorithm}

    We initiate with $S = \emptyset$.
    The first loop iterates as long as there exists $X \in \partition \setminus S$
    that is weakly $r$-reachable in $G-\bigcup S$
    from at least $|A|/(2mc)$ elements of $A$. 
    Intuitively, this $X$ is a ``hub'' and should be put into $S$. 
    To limit the number of iterations of the first loop, we restrict $A$
    to only those $Y \in A$ from which $X$ is weakly $r$-reachable in $G-\bigcup S$. 
    
    Observe that at any point of the first loop,
    we have $S \subseteq \bigcap_{Y \in A} \WReachS{Y}$ and $A \neq \emptyset$.
    Thus, $|S| \leq c$ holds at any point of the algorithm,
    and the first loop iterates at most $c$ times.
    
    In every iteration of the first loop, the size of $A$ shrinks at most by a factor of $2mc$. 
    Hence, if initially $|A| \geq (2mc)^{c+1}$, then once we exit the first loop, 
    $|A| \geq 2mc \geq 2c$. Consequently, $A' := A \setminus S$ is of size
    at least $|A|-c \geq |A|/2$.

    In the second loop, we greedily construct a subset $B \subseteq A'$ 
    such that the sets $\WReach{\mainR}{G-\bigcup S}{\partition\setminus S}{\preceq}{X}$ for $X \in B$ are pairwise disjoint. 
    
    Let $X \in B$ and consider the iteration of the second loop during which $X$ was added to $B$. At the start of that iteration, we had $X \in A'$ so $X \notin S$. 
    Furthermore, if $Z \in \WReach{\mainR}{G - \bigcup S}{\partition\setminus S}{\preceq}{X}$, then there are at most $|A|/(2mc)$ elements $Y \in A$ such that $Z \in \WReach{\mainR}{G - \bigcup S}{\partition\setminus S}{\preceq}{Y}$ since the first loop has ended.
    Since $|\WReach{\mainR}{G - \bigcup S}{\partition\setminus S}{\preceq}{X}| \leq \wcol_{\mainR}(G, \partition, \preceq) = c$, there are at most $c \cdot |A|/(2mc)$ sets $Y \in A'$ such that $\WReach{\mainR}{G - \bigcup S}{\partition\setminus S}{\preceq}{X} \cap \WReach{\mainR}{G - \bigcup S}{\partition\setminus S}{\preceq}{Y} \neq \emptyset$, and $X$ is such a set.
    Thus, when we added $X$ to $B$, we removed at most $|A|/2m$ elements from $A'$, including $X$.
    Since initially $|A'| \geq |A|/2$ and we only stop when $A'$ is empty, we have $|B| \geq m$.
    
    Finally, let $X_1 \neq X_2 \in B$ and assume without loss of generality that $X_1$ was added to $B$ before $X_2$.
    By contradiction, suppose that there is a path $P$ of length at most $\mainR$ from $X_1$ to $X_2$ in $G - \bigcup S$. Let $Y \in \partition$ be the $\preceq$-minimum element that contains a vertex of $P$. Then, 
    $Y \in \WReach{\mainR}{G - \bigcup S}{\partition\setminus S}{\preceq}{ X_1} \cap \WReach{\mainR}{G -\bigcup S}{\partition\setminus S}{\preceq}{X_2}$, as witnessed respectively by the prefix and the suffix of $P$.
    Thus, if $X_2$ was in $A'$ when $X_1$ was added to $B$, $X_2$ was removed from $A'$ at that point, so $X_2$ cannot have been added to $B$ after $X_1$, a contradiction. 
\end{proof}

We conclude this section with verifying \cref{cor:flat}.

\corflat*

\begin{proof}
    First, apply the algorithm of \cref{thm:wcol} to $G$ and $\tinyR$, obtaining
    a partition $\partition$ and ordering $\preceq$. 
    Observe that, since every element of $\partition$ is of strong diameter 
    at most $\tinyR$ while $A$ is $\tinyR$-scattered,
    every element of $\partition$ contains at most one element of $A$
    and $A_\partition := \{ X \in \partition~|~X \cap A \neq \emptyset\}$ is
    of size at least $(2mc)^{c+1}$. 
    Apply \cref{thm:wcol2flat} to $G$, $\mainR$, $\partition$, $\preceq$, $m$,
    and $A_\partition$, obtaining sets $S \subseteq \partition$ and
    $B_\partition \subseteq A_\partition \setminus S$.
    Since every element of $A_{\partition}$ contains exactly one vertex of $A$,
    $B := A \cap \bigcup B_\partition$ is of size at least $m$ and, 
    as every two distinct sets of $B_\partition$ are at distance more than $\mainR$
    in $G-\bigcup S$, $B$ is $\mainR$-scattered in $G-\bigcup S$. 
    Thus, $S$ and $B$ are as promised.
\end{proof}
\section{Bounding the $\varepsilon$-scatter dimension}\label{sec:ladders}
With \cref{thm:wcol,thm:wcol2flat}, we can now conclude 
the proof of \cref{thm:dim-bound}. 
The following lemma is the last step.

\begin{lemma} \label{lem:no-long-scatter}
    Let $G$ be an edge-weighted graph, let $\varepsilon > 0$,
    and let $\mainR > \tinyR > 0$ be such that $\tinyR \leq \frac{\varepsilon}{3} \cdot \mainR$.
    Let $\partition$ be a partition of $V(G)$ whose every part has
    weak diameter at most $\tinyR$ 
    and $\preceq$ be an ordering of $\partition$ with $c := \wcol_{3\mainR}(G,\partition,\preceq)$. 
    Then, every $\varepsilon$-ladder of width $\mainR$ in $G$ has length less than
    \[ \ell := \left(6c \cdot \left(\frac{\mainR}{\tinyR} + 2\right)^c\right)^{c+1}. \]
\end{lemma}

\begin{proof}
    Let $m = 3 \cdot \left(\frac{\mainR}{\tinyR} + 2\right)^c$ so that $\ell = (2cm)^{c+1}$.
    By contradiction, suppose that $(x_1, p_1), \ldots, (x_{\ell}, p_{\ell})$ is an $\varepsilon$-ladder of width $\mainR$ in $G$.

    For every $i \in [\ell]$, let $X_i \in \partition$ be the unique set in $\partition$ such that $x_i \in X_i$.
    Let $A = \{X_1, \ldots, X_{\ell}\} \subseteq \partition$.
    Since for every $1 \leq i < j \leq \ell$ we have $\dist_G(x_i,p_i) > (1+\varepsilon)\mainR$
    but $\dist_G(x_j,p_i) \leq \mainR$, we have $\dist_G(x_i,x_j) > \varepsilon \mainR \geq 3\tinyR$. Hence, $\{x_i, i \in [\ell]\}$ is $3\tinyR$-scattered in $G$
    and, consequently, the sets $X_i$ are pairwise distinct. 

    We apply \cref{thm:wcol2flat} to $G$, $3\mainR$, $\mathcal{P}$, $\preceq$, $m$ and $A$: There exist $S \subseteq \partition$ and $B \subseteq A \setminus S$ such that $|S| \leq c, |B| \geq m$ and every two distinct $X_i, X_j \in B$ are at distance at least $3\mainR$ in $G - \bigcup S$.
    Write $S = \{S_1, \ldots, S_t\}$.

    We define for every $X \in B$ its \emph{distance profile} to $S$, as follows.
    For $X \in B$ and $S_i \in S$, set
    \[ \delta(X, S_i) := \begin{cases} \left\lfloor \frac{\dist_G(X, S_i)}{\tinyR} \right\rfloor  & \mathrm{if}\ \dist_G(X, S_i) \leq \mainR,\\ +\infty & \mathrm{otherwise} \end{cases} \]
    The distance profile of $X$ to $S$ is $\mathrm{prof}_S(X) = (\delta(X, S_1), \ldots, \delta(X, S_t))$.
    Observe that there are at most $(\frac{\mainR}{\tinyR} + 2)^c$ possible distance profiles to $S$.
    By the pigeonhole principle, there exist $i < j < k$ such that $X_i, X_j, X_k \in B$ and $\mathrm{prof}_S(X_i) = \mathrm{prof}_S(X_j) = \mathrm{prof}_S(X_k)$.

    Let $Q_j$ be a shortest path in $G$ from $x_j$ to $p_i$ and $Q_k$ a shortest path in $G$ from $x_k$ to $p_i$. Since $(x_1, p_1), \ldots, (x_{\ell}, p_{\ell})$ is an $\varepsilon$-ladder in $G$ of width $\mainR$, then both $Q_j$ and $Q_k$ have length at most $\mainR$, so their concatenation is a path of length at most $2\mainR$ in $G$ between $X_j$ and $X_k$. 
    Since $X_j$ and $X_k$ are at distance more than $3\mainR$ in $G-\bigcup S$, then at least one of $Q_j, Q_k$ intersects $\bigcup S$, say $Q_j$ without loss of generality.
    Let $S_l \in S$ be such that $Q_j$ intersects $S_l$ and $s \in S_l \cap V(Q_j)$. 
    Then $\dist_G(X_j, S_l) \leq \mainR$, and $\mathrm{prof}_S(X_i) = \mathrm{prof}_S(X_j)$ implies $\dist_G(X_i, S_l) \leq \dist_G(X_j, S_l) + \tinyR$.
    Using that $X_i$ and $S_l$ have diameter at most $\tinyR$, we therefore obtain 
    \begin{align*} 
    \dist_G(x_i, p_i) &\leq \tinyR + \dist_G(X_i, S_l) + \tinyR + \dist_G(s, p_i) \\
    &\leq \dist_G(X_j, S_l) + \dist_G(s, p_i) + 3\tinyR \\
    &\leq \dist_G(x_j, s) + \dist_G(s, p_i) + \varepsilon \mainR \\
    &= \dist_G(x_j, p_i) + \varepsilon \mainR \leq (1 + \varepsilon)\mainR.
    \end{align*}
    This is the final contradiction.
\end{proof}

\cref{thm:dim-bound} follows directly by plugging in the 
pair $(\partition, \preceq)$ obtained from \cref{thm:wcol}
for $G$ and $\tinyR := \frac{\varepsilon}{3} \mainR$ for arbitrary $\mainR > 0$ into \cref{lem:no-long-scatter}.

\section{Lower bound}\label{sec:lb}
This section is devoted to the proof of \cref{thm:dim-lb}.

For integers $k,r,d \geq 1$, we define inductively a graph $G(k,r,d)$ 
and a matching $M(k,r,d) \subseteq E(G(k,r,d))$ whose endpoints we call henceforth \emph{the set of twin leaf vertices} of $G(k,r,d)$. 
For a vertex $v$, by \emph{creating a twin of $v$} we mean the following operation: we create a new vertex $v'$ that is adjacent exactly to $v$ and all neighbors of $v$. 
\begin{itemize}
    \item $G(1, r, d)$ is built from the complete rooted $d$-ary tree of depth $r+1$ (i.e., any path from leaf to root has $r$ edges and $r+1$ vertices) by creating a twin $v'$ for every leaf $v$ of the tree; the edges $vv'$ over all leaves $v$ form the matching $M(1,r,d)$.
    \item $G(k, 1, d)$ is built from the complete rooted $d$-ary tree of depth $k+1$ by adding an edge $uv$ whenever $u$ and $v$ are in the ancestor-descendant relation and then creating a twin $v'$ for every leaf $v$ of the tree; the edges $vv'$ over all leaves $v$ form the matching $M(k,1,d)$.
    \item $G(k+1, r+1, d)$ is built as follows: start from a copy of $G(k+1, r, d)$ and for every edge $vv'$ of $M(k+1,r,d)$, create $d$ disjoint copies of $G(k, r+1, d)$ and add an edge between every vertex from a copy of $G(k,r+1,d)$ and both $v$ and $v'$. 
    The matching $M(k+1,r+1,d)$ is the union of all matchings $M(k,r+1,d)$ of
    the copies of $G(k,r+1,d)$ used in the construction.
\end{itemize}

\begin{lemma}\label{lem:lb-tw}
    For every $d, k, r \geq 1$, the graph $G(k, r, d)$ has treewidth at most $2k$.
\end{lemma}

\begin{proof}
    We prove it by induction on $k+r$.
    
    For $k = 1$, recall that the complete $d$-ary tree of depth $r+1$ has a tree decomposition of width 1 where each bag contains at most one leaf.
    Hence, $G(1, r, d)$ has a tree decomposition with all bags of size at most 3, thus $\treewidth(G(1, r, d)) \leq 2$.
    
    For $r = 1$, recall that $G(k, 1, d)$ is built from the complete $d$-ary tree of depth $k+1$ by adding an edge $uv$ whenever $u$ and $v$ are comparable in the tree order and then creating a twin for every leaf of the tree. Let $H$ be the graph we have before creating the twins. Observe that $H$ has a path decomposition where each bag has size at most $k+1$ and contains exactly one leaf of the tree. Therefore, $G(k, 1, d)$ has treewidth at most $k+1 \leq 2k$.
    
    For the inductive step, consider $G(k+1, r+1, d)$ and assume the property holds for $G(k+1, r, d)$ and $G(k, r+1, d)$. Take a tree decomposition of $G(k+1, r, d)$ of width at most $2k + 2$ and observe that every pair of twin leaf vertices are adjacent, thus they appear in some common bag. For every copy of $G(k,r+1,d)$, let $v,v'$ be the two vertices of
        $G(k+1,r,d)$ it is adjacent to, take a tree decomposition of $G(k, r + 1, d)$
        with width at most $2k$, add $v$ and $v'$ to every bag of this tree decomposition
        and make its root adjacent to any bag of the tree decomposition of $G(k+1,r,d)$ that contains both $v$ and $v'$.
\end{proof}

\begin{lemma}\label{lem:lb-ell}
    For every~$0 < \varepsilon < 1$, 
    and integers $k,r \geq 1$ such that $r < \frac{1-\varepsilon}{\varepsilon}$,
    there exists an assignment of weights to the edges of $G(k,r,2)$ such that
    the resulting edge-weighted graph contains an $\varepsilon$-ladder $(x_i,p_i)_{i=1}^\ell$ of
    length $2^{\binom{k+r-2}{r-1}}$ and width $1$
    with the following additional properties: \begin{itemize}
        \item For every $i \in [\ell]$, $x_i p_i \in M(k,r,2)$.
        \item For every $i, j \in [\ell]$, $\dist(x_j, p_i) \geq 1$.
    \end{itemize}
\end{lemma}

\begin{proof}
    We prove it by induction on $k+r$. 
    
    For $k = 1$, recall that $G(1, r, 2)$ is obtained from the complete binary tree of depth $r+1$ by creating a twin for every leaf. Pick two arbitrary pairs of twin leaf vertices $(x, p)$ and $(x', p')$. 
    Give weight $\infty$ to all edges incident to $x$ and $p'$, weight $1/2$ to the edge from $p$ to its parent and to the edge from $x'$ to its parent, and weight 0 to all other edges. Then, $(x, p), (x', p')$ is an $\varepsilon$-ladder of length 2 and width 1 with the desired properties.
        
    For $r = 1$, recall that $G(k, 1, 2)$ is a supergraph of the complete binary tree of depth $k+1$ for which we created a twin for every leaf. Pick two arbitrary pairs of twin leaf vertices $(x, p)$ and $(x', p')$. 
    Give weight $\infty$ to all edges incident to $x$ and $p'$, weight $1/2$ to the edge from $p$ to its parent and to the edge from $x'$ to its parent, and weight 0 to all other edges. Then, $(x, p), (x', p')$ is an $\varepsilon$-ladder of length 2 and width $1$ with the desired properties.

    For the inductive step, consider $G(k+1, r+1, 2)$ and assume the property holds for $G(k+1, r, 2)$ and $G(k, r+1, 2)$.
    Recall that we assume that $0 < \varepsilon < 1$ satisfies $r +1 < \frac{1-\varepsilon}{\varepsilon}$.
        Let 
        \[ \gamma = \frac{1-\varepsilon}{2(r+1)\varepsilon} > \frac{1}{2}\quad\mathrm{and}\quad \varepsilon' = \frac{\varepsilon}{1-2\gamma\varepsilon}. \]
        First, note that
        \[ \gamma = \frac{1-\varepsilon}{2(r+1)\varepsilon} < \frac{1-\varepsilon}{2\varepsilon} \ \Longrightarrow\ 0 < \varepsilon < 1-2\gamma\varepsilon\ \Longrightarrow\ 0 < \varepsilon' < 1. \]
        Note also that 
        \[ \frac{\varepsilon'}{1-2\gamma r \varepsilon'} = \frac{\varepsilon}{1-2\gamma(r+1)\varepsilon} = \frac{\varepsilon}{1-(1-\varepsilon)} = 1\ \Longrightarrow\ 
        r < 2\gamma r = \frac{1-\varepsilon'}{\varepsilon'}. \]

    By the inductive hypothesis, there exists an assignment of weights
    to $G(k+1,r,2)$ and a $\varepsilon'$-ladder $(x_i,p_i)_{i=1}^{\ell_1}$ in $G(k+1,r,2)$ that satisfies the additional properties from the lemma, 
    with $\ell_1 = 2^{\binom{k+r-1}{r-1}}$.
    We scale the weights of $G(k+1,r,2)$ by a factor $(1 - 2\gamma\varepsilon)$.
    
    Recall that for every $i \in [\ell_1]$, $x_ip_i \in M(k+1,r,2)$
    and $G(k+1,r+1,2)$ features two copies of $G(k,r+1,2)$ complete to $x_i$ and $p_i$. 
    Pick one such copy; by the inductive assumption, this copy admits edge weights
    with an $\varepsilon$-ladder $(y^i_j,p^i_j)_{j=1}^{\ell_2}$
    that satisfies the conditions of the lemma
    and has length $\ell_2 = 2^{\binom{k+r-1}{r}}$. Assign such edge weights to this copy
    of $G(k,r+1,2)$ without any scaling. 
    
    Set the weight of the edge between $y^i_j$ and $x_i$ to $\gamma\varepsilon$, same for the edge between $p^i_j$ and $p_i$, and set all other weights of edges between vertices of $M(k+1,r,2)$ and copies of $G(k,r+1,2)$ to $\infty$.
    Set also to $\infty$ all weights of edges of copies of $G(k,r+1,2)$ that
    were not assigned earlier.

    We claim that the sequence 
    \[ (y^1_1, p^1_1), (y^1_2, p^1_2), \ldots, (y^1_{\ell_2}, p^1_{\ell_2}), (y^2_1, p^2_1), \ldots, (y^{\ell_1}_{\ell_2}, p^{\ell_1}_{\ell_2}) \]
    is an $\varepsilon$-ladder in $G(k+1, r+1, 2)$ which satisfies the conditions of the lemma.
    Write the sequence as $(z_1, q_1), \ldots, (z_\ell, q_\ell)$.
    
        Let $i < j \in [\ell]$. If $z_j$ and $q_i$ are in the same copy of $G(k, r+1, 2)$ then by induction assumption there exists a path from $z_j$ to $q_i$ of length at most $1$ inside this copy of $G(k, r+1, 2)$,
        so $\dist(z_j, q_i) \leq 1$.
        If not, let $x_J,p_J$ be the attachment points of the copy of $G(k, r+1, 2)$ that contains $z_j$ and $x_I,p_I$ that of the copy of $q_i$.
        As $i < j$, we have $I < J$, hence there is a path of length at most 1 from $x_J$ to $p_I$ in the copy of $G(k+1, r, 2)$. Since we scaled down the weights, this path has now length at most $(1-2\gamma\varepsilon)$. Therefore, 
        \[ \dist(z_j, q_i) \leq \dist(z_j, x_J) + \dist(x_J, p_I) + \dist(p_I, q_i) \leq \gamma\varepsilon + 1-2\gamma\varepsilon + \gamma\varepsilon \leq 1. \]

        We now move to lower bounding some distances. To this end, we need the following observation:
        \begin{equation}\label{eq:lb-upper-dist}
            \forall_{I \in [\ell_1]} \dist(x_I,p_I) > 1+\varepsilon-2\gamma\varepsilon.
        \end{equation}
        To prove~\eqref{eq:lb-upper-dist}, consider a shortest path $P$ from $x_I$ to $p_I$. 
        If $P$ does not visit any copy of $G(k,r+1,2)$, then it stays inside the 
        copy of $G(k+1,r,2)$ and its length is more than
        \[ (1+\varepsilon')(1-2\gamma\varepsilon) = 1+\varepsilon -2\gamma\varepsilon. \]
        Otherwise, $P$ visits a copy $H_J$ of $G(k,r+1,2)$. It enters and exits $H_J$ via
        $x_J$ and $p_J$; hence, it traverses $H_J$ from some $x^J_i$ to some $p^J_j$. 
        By the inductive hypothesis, $\dist_{H_J}(x^J_i, p^J_j) \geq 1 > 1 + \varepsilon -2\gamma\varepsilon$, as $\gamma > \frac{1}{2}$. This proves~\eqref{eq:lb-upper-dist}. 
        
        Now, let $i \in [\ell]$ and consider a shortest path $P$ from $z_i$ to $q_i$, using as few vertices as possible. Let $H_i$ be the copy of $G(k, r+1, 2)$ that contains $z_i$ and $q_i$ and let $x, p$ be its attachment points in $G(k+1,r,2)$.
        If $P$ visits neither $x$ nor $p$ then $P$ stays in $H_i$ so by induction assumption has length more than $1+\varepsilon$. 
        If $P$ visits $x$ but not $p$ then $P$ stays in $H_i \cup \{x\}$ so the vertex visited after $x$ is some $z_j \in V(H_i)$ (as other edges have infinite weight) and after this $P$ stays inside $H_i$. Thus, $P$ has length at least 
        \[\dist(z_i, x) + \dist(x, z_j) + \dist_{H_i}(z_j, q_i) \geq \gamma\varepsilon + \gamma\varepsilon + 1 > 1+\varepsilon,\] where we used the induction assumption to argue that $d_{H_i}(z_j, q_i) \geq 1$, and the fact that $\gamma > 1/2$.
        The same reasoning holds if $P$ visits $p$ but not $x$.
        Finally, suppose that $P$ visits both $x$ and $p$. Observe that $p$ is at distance exactly $\gamma\varepsilon$ from $q_i$ while $x$ is at distance $\gamma\varepsilon$ from $H_i$ so at distance at least $\gamma\varepsilon$ from $q_i$. Thus, by minimality of $P$, $P$ does not visit $x$ after visiting $p$. Therefore, using also~\eqref{eq:lb-upper-dist}, $P$ has length 
        \[ \dist(z_i, x) + \dist(x, p) + \dist(p, q_i) = 2\gamma\varepsilon + \dist(x, p) > 1+\varepsilon.\]

        We now prove that for every $i, j \in [\ell]$ we have $\dist(z_i, q_j) \geq 1$.
        If $z_i$ and $q_j$ are not in the same copy of $G(k, r+1, 2)$, let $(x_I, p_I)$ be the attachment points of the copy $H_I$ of $G(k,r+1,2)$ containing $z_i$, and $(x_J, p_J)$ be the attachment points of the copy $H_J$ containing $q_j$.
        Let $P$ be a shortest path from $z_i$ to $q_j$. Then, $P$ visits at least one of $\{x_I, p_I\}$ and at least one of $\{x_J, p_J\}$. 
        Consider the first vertex of $\{x_I, p_I\}$ which $P$ visits. If it is $p_I$ then the previous vertex was some $p^I_{i'} \in V(H_I)$, and by induction assumption $d_{H_I}(z_i, p^I_{i'}) \geq 1$ so $P$ has length at least 1.
        Similarly, if the last vertex of $\{x_J, p_J\}$ which $P$ visits is $x_J$ then the next vertex is some $y^J_{j'} \in V(H_J)$, and by induction assumption $d_{H_J}(y^J_{j'}, q_j) \geq 1$ so $P$ has length at least 1.
        Lastly, if $P$ visits both $x_I$ and $p_J$ then $P$ has length $2\gamma\varepsilon + \dist(x_I, p_J)$. By induction assumption, in the original copy of $G(k+1, r, 2)$, the distance between $x_I$ and $p_J$ was at least 1, so their distance in the copy of $G(k+1, r, 2)$ is now at least $1-2\gamma\varepsilon$.
        Furthermore, any path from $x_I$ to $p_J$ that visits some $H_{I'}$ needs to contain a subpath from $x_{I'}$ to $y_{I'}$ that, by~\eqref{eq:lb-upper-dist}, is longer than $1-2\gamma\varepsilon$. Thus, 
        $\dist(x_I, p_J) \geq 1-2\gamma\varepsilon$ so $P$ has length at least 1.
        
        In the other case, if $z_i$ and $q_j$ are in the same copy $H$ of $G(k, r+1, 2)$ with attachment points $(x, p)$, again let $P$ be a shortest path from $z_i$ to $q_j$. If $P$ visits neither $x$ nor $p$ then $P$ stays inside $H$ so has length at least 1 by induction assumption. 
        If the last vertex of $P$ outside $H$ is $x$, then the next vertex is some $z_{i'} \in V(H)$
        and the rest of $P$ stays inside $H$, and $\dist_H(z_{i'}, q_j) \geq 1$ by induction hypothesis so $P$ has length at least 1. 
        A similar argument holds if the first vertex of $P$ outside $H$ is $p$. 
        In the remaining case, $P$ visits both $x$ and $p$, and visits $x$ before $p$.
        Then, by~\eqref{eq:lb-upper-dist}, the length of $P$ is at least
        $\gamma\varepsilon + 1 + \varepsilon - 2\gamma\varepsilon + \gamma\varepsilon = 1+\varepsilon$. This finishes the proof that $\dist(z_i,q_j) \geq 1$. 

        To conclude the proof of the lemma, note that 
        \[ \ell = \ell_1\ell_2 = 2^{\binom{k+r-1}{r-1}} \cdot 2^{\binom{k+r-1}{r}} = 2^{\binom{k+r}{r}}.\]
\end{proof}

We are now ready to conclude the proof of \cref{thm:dim-lb}.

\thmdimlb*

\begin{proof}
    Observe that $\varepsilon < \frac{1}{r+2}$ implies $r + 1 < \frac{1-\varepsilon}{\varepsilon}$.
    Hence, $G(k+1,r+1,2)$ is of treewidth at most $2k+2$ by \cref{lem:lb-tw}
    and, by \cref{lem:lb-ell}, admits edge weights with an $\varepsilon$-ladder of length at least $2^{\binom{k+r}{k}}$, as desired.
\end{proof}
\section{A coreset for \textsc{$k$-Center}}\label{sec:kcenter}
This section is devoted to the proof of \cref{thm:coreset}.
The following statement contains precise size bounds.

\begin{theorem} \label{thm:coreset-full}
Given a \textsc{$k$-Center} instance with the metric given as a graph metric
on a $K_h$-minor-free graph $G$, and an accuracy parameter $0 < \varepsilon < 1$
one can in polynomial time compute a subset $S$ of clients with the following guarantee:
for every set $X$ of at most $k$ potential centers, it holds that
\[ \max_{\mathrm{client}\ p} \dist(p,X) \leq (1+\varepsilon) \max_{p \in S} \dist(p,X).\]
Furthermore, the size of $S$ is bounded by
\[ k + \frac{9}{\varepsilon^3} \cdot \left(2c(\Gamma \cdot k+2) \cdot \left(\frac{3}{\varepsilon} + 2\right)^c\right)^{c+1}
\] 
where $c = c(h, 9/\varepsilon)$ where $c(\cdot,\cdot)$ comes from \cref{thm:wcol}
and 
\[ \Gamma = 3(c+1)\left[2k+c(4 + 3/\varepsilon + 2k)\right]. \]
\end{theorem}

\begin{proof}
    Start by computing a 2-approximate solution $\tilde{X}$, e.g., by the greedy
    algorithm of~\cite{Gonzalez85}.
    Let $\tilde{\mainR} := \max_{\mathrm{client}\ p} \dist(p, \tilde{X})$ be the 
    value of $\tilde{X}$. 

    Compute a maximal $2\tilde{\mainR}$-scattered set $Z$ of clients. 
    Note that $|Z| \leq k$.
    
    Let $\delta = \sqrt{1 + \varepsilon} - 1$ so that $(1 + \delta)^2 = 1 + \varepsilon$.
    Let $\displaystyle \mainR^\ast = \frac{\tilde{\mainR}}{2(1 + \varepsilon)}$.
    For every integer $t \geq 0$ such that $(1+\delta)^t \cdot \mainR^\ast \leq 2(1+\delta)\tilde{\mainR}/\varepsilon$, we define a set $P_t$ as follows.

    Let $\mainR_t =(1+\delta)^t \cdot \mainR^\ast$. Compute a sequence $(X_1, p_1), (X_2, p_2), \ldots, (X_{\ell}, p_{\ell})$ where each $p_i$ is a client and each $X_i$ is a subset of potential centers with the following properties: \begin{itemize}
        \item $\dist(p_i, X_i) > (1 + \delta)\mainR_t$ for every $i \in [\ell]$.
        \item $\dist(p_i, X_j) \leq \mainR_t$ whenever $i < j$.
        \item $|X_i| \leq \lambda \cdot k$ for every $i \in [\ell]$, where 
        $\lambda$ is defined by the equation
        \[ \lambda = \log\left[\left(2c(\lambda \cdot k+2) \cdot \left(\frac{3}{\varepsilon} + 2\right)^c\right)^{c+1}\right]. \]
        \item We could not extend this sequence with a pair $(X_{\ell + 1}, p_{\ell + 1})$ such that the first two bullets hold and $|X_{\ell + 1}| \leq k$.
    \end{itemize}
    Finally, set $P_t = \{p_1, \ldots, p_{\ell}\}$ and set
    $S = Z \cup \bigcup_{t \geq 0} P_t$.
    We claim that $S$ satisfies the desired properties. 

    \begin{claim}
    For every set $X$ of at most $k$ potential centers, 
    \[ \max_{\mathrm{client}\ p} \dist(p, X) \leq (1+\varepsilon) \max_{p \in S} \dist(p, X).\]
    \end{claim}

    \begin{subproof}
    Let $d = \max_{p \in S} \dist(p, X)$.
    Suppose first that $d > 2\tilde{\mainR} / \varepsilon$.
    Pick a client $p$. 
    By maximality of $Z$, there exists $p' \in Z$ within distance $2\tilde{\mainR}$ from
    $p$. Hence, $\dist(p, X) \leq d + 2\tilde{\mainR} \leq (1+\varepsilon)d$, as desired.

    Assume then $d \leq 2\tilde{\mainR} / \varepsilon$.
    By contradiction, assume there exists a client $p$ with $\dist(p, X) > (1+\varepsilon)d$.
    
    Let $t \geq 0$ be minimum such that $\mainR_t \geq d$ (such a $t$ exists since the last $r_t$ we consider is at least $2\tilde{\mainR} / \varepsilon \geq d$). 
    We claim that there exists a client $q$ with $\dist(q, X) > (1+\delta)\mainR_t$.
    If $t = 0$, then $\mainR_t = \tilde{\mainR}/[2(1 + \varepsilon)]$
    and the value of the optimum solution is at least $\tilde{\mainR}/2$, hence
    there exists a client $q$ such that $\dist(q, X) \geq (1 + \varepsilon)\mainR_t > (1 + \delta)\mainR_t$.
    If $t > 0$, by minimality of $t$, $d > \mainR_{t-1} = \mainR_t/(1+\delta)$.
    Hence, $\dist(p, X) > (1+\delta)\mainR_t$, so we can take $q=p$. 

    Consider the sequence we computed for $\mainR_t$, denote it by $(X_1, p_1), \ldots, (X_{\ell}, p_{\ell})$. By construction, for every $i \in [\ell]$, $p_i \in P_t \subseteq S$ so $\dist(p_i, X) \leq d \leq \mainR_t$. Furthermore, $\dist(q, X) > (1 + \delta)\mainR_t$ and $|X| \leq k$ so we could have extended the sequence with $(X, q)$, which is a contradiction. This finishes the proof of the claim.
    \end{subproof}

    \begin{claim} \label{cl:length-k-ladder}
        For every $t$, the sequence computed for $\mainR_t$ has length at most $\left(2c(\lambda \cdot k+2) \cdot (\frac{3}{\varepsilon} + 2)^c\right)^{c+1}$.
    \end{claim}

    \begin{subproof}
        This proof is a simple adaptation of the proof of \cref{lem:no-long-scatter}.
        Let $\tinyR = \varepsilon \cdot \mainR_t / 3, \mainR = 3\mainR_t$ and 
        note that $\frac{\mainR}{\tinyR} = \frac{9}{\varepsilon}$. 
        Set $m = (\lambda \cdot k+2) \cdot (\frac{3}{\varepsilon} + 2)^c$ so that $\left(2c(\lambda \cdot k+2) \cdot (\frac{3}{\varepsilon} + 2)^c\right)^{c+1} = (2cm)^{c+1}$
    
        By contradiction, suppose the sequence $(X_1, p_1), \ldots, (X_{\ell}, p_{\ell})$ computed for $\mainR_t$ has length $\ell \geq (2cm)^{c+1}$.

        By \cref{thm:wcol}, there exists a partition $\mathcal{P}$ of $V(G)$ such that every $Y \in \mathcal{P}$ has strong diameter 
        at most $\tinyR$, and a total order $\preceq$ on $\partition$ such that $\wcol_{\mainR}(G, \mathcal{P}, \preceq) \leq c$. 
        For every $i \in [\ell]$, let $Y_i \in \mathcal{P}$ be such that $p_i \in Y_i$, and let $A = \{Y_1, \ldots, Y_{\ell}\}$. Note that since $\dist(p_i, p_j) > \varepsilon \cdot \mainR_t$ whenever $i \neq j$ then $Y_i \neq Y_j$ whenever $i \neq j$ so $|A| \geq (2mc)^{c+1}$.
        By \cref{thm:wcol2flat}, there exist $S \subseteq \mathcal{P}$ and $B \subseteq A \setminus S$ such that $|S| \leq c, |B| = m$ and every two distinct $Y_i, Y_j \in B$ are at distance at least $\mainR$ in $G -\bigcup S$.

        We define for every $Y \in B$ its \emph{distance profile} to $S$, as follows.
        For $S_i \in S$, set $\delta(Y, S_i) = \lfloor \frac{\dist_G(Y, S_i)}{\tinyR} \rfloor$ if $\dist_G(Y, S_i) \leq \mainR_t$,
        and $\delta(Y, S_i) = \infty$ otherwise.
        Writing $S = \{S_1, \ldots, S_l\}$, the distance profile of $Y$ to $S$ is $\mathrm{prof}_{S}(Y) = (\delta(Y, S_1), \ldots, \delta(Y, S_l))$.
        Observe that there are at most $(\frac{3}{\varepsilon} + 2)^c$ possible distance profiles to $S$.
        By the pigeonhole principle, there exist $j_1 < j_2 < \ldots < j_{\lambda \cdot k+2}$ such that $Y_{j_1}, Y_{j_2}, \ldots, Y_{j_{\lambda \cdot k+2}} \in B$ and $\mathrm{prof}_{S}(Y_{j_1}) = \mathrm{prof}_{S}(Y_{j_2}) = \ldots = \mathrm{prof}_{S}(Y_{j_{\lambda \cdot k+2}})$.

        For $i \in [\lambda \cdot k+1]$, let $Q_{j_i}$ be a shortest path in $G$ from $X_{j_{\lambda \cdot k+2}}$ to $p_{j_i}$.
        Since there are $\lambda \cdot k+1$ such paths, and since $|X_{j_{\lambda \cdot k+2}}| \leq \lambda \cdot k$, two of them have the same endpoint in $X_{j_{\lambda \cdot k+2}}$, call it $x$.
        Both of them have length at most $\mainR_t$ by definition of the sequence so their concatenation forms a path of length at most $2\mainR_t < \mainR$, between two distinct sets in $B$. Therefore, one of them has to intersect $S$, say the path $Q_{j_i}$ to $p_{j_i}$. Let $S_{i'} \in S$ be such that $S_{i'}$ intersects $Q_{j_i}$, say at $s \in S_{i'}$.
        Then $\dist_G(Y_{j_{i}}, S_{i'}) \leq \mainR_t$, and $\mathrm{prof}_S(Y_{j_i}) = \mathrm{prof}_S(Y_{j_{\lambda \cdot k + 2}})$ implies $\dist_G(Y_{j_{\lambda \cdot k + 2}}, S_{i'}) \leq \dist_G(Y_{j_i}, S_{i'}) + \tinyR$.
        Using that $Y_{j_{\lambda \cdot k + 2}}$ and $S_{i'}$ have diameter at most $\tinyR$, we therefore obtain 
        \begin{align*}    
        \dist_G(p_{{j_{\lambda \cdot k + 2}}}, X_{j_{\lambda \cdot k + 2}}) &\leq \tinyR + \dist_G(Y_{j_{\lambda \cdot k + 2}}, S_{i'}) + \tinyR + \dist_G(s, x) \\&\leq \dist_G(Y_{j_i}, S_{i'}) + \dist_G(s, x) + 3\tinyR \\&\leq \dist_G(p_{j_i}, s) + \dist_G(s, x) + \varepsilon \mainR_t \\&= \dist_G(p_{j_i}, x) + \varepsilon \mainR_t \leq (1+\varepsilon)\mainR_t.
        \end{align*}
        This is a contradiction.
    \end{subproof}

    \begin{claim} \label{cl:number-k-ladders}
         \[ |S| \leq k + \frac{9}{\varepsilon^3} \cdot \left(2c\left(\lambda \cdot k+2\right) \cdot \left(\frac{3}{\varepsilon} + 2\right)^c\right)^{c+1}. \]
    \end{claim}

    \begin{subproof}
        We saw that $|Z| \leq k$, so by \cref{cl:length-k-ladder}, it suffices to prove that we compute at most $\frac{9}{\varepsilon^3}$ sequences, i.e. that $(1 + \delta)^\frac{9}{\varepsilon^3}  \cdot \mainR^\ast > 2(1 + \delta)\tilde{\mainR}/\varepsilon$.
        \begin{align*}
            (1 + \delta)^M  \cdot \mainR^\ast > 2(1 + \delta)\tilde{\mainR}/\varepsilon &\iff (1 + \delta)^{M-2} \cdot \tilde{\mainR} > 4(1 + \delta) \tilde{\mainR} / \varepsilon \\
                &\iff (1 + \delta)^{M - 3} > 4/\varepsilon \\
                &\iff (M - 3) \log(1 + \delta) > 2 + \log(1 / \varepsilon) \\
                &\iff (M - 3) \cdot \frac{1}{2} \log(1 + \varepsilon) > 2 + \log(1 / \varepsilon) \\
                &\iff M > 3 + \frac{2}{\log(1 + \varepsilon)}(2 + \log(1 / \varepsilon))
        \end{align*}
        However, $\log(1/\varepsilon) \leq 1 / \varepsilon$ and $\log(1 + \varepsilon) \geq \varepsilon^2$ so $3 + \frac{2}{\log(1 + \varepsilon)}(2 + \log(1 / \varepsilon)) \leq 3 + \frac{2}{\varepsilon^2}(2 + \frac{1}{\varepsilon})$. Finally, using that $\varepsilon < 1$, we get $3 + \frac{2}{\log(1 + \varepsilon)}(2 + \log(1 / \varepsilon)) < \frac{9}{\varepsilon^3}$ so we indeed compute at most $\frac{9}{\varepsilon^3}$ sequences, which concludes the proof.
    \end{subproof}

    To obtain the promised bound for $|S|$, note that $\lambda \leq \Gamma$.

    Finally, we verify that the sequences $(X_i,p_i)$ can be computed in polynomial time.
    To this end, we use a greedy algorithm for set cover.

    \begin{claim} \label{cl:extend-sequence}
        Given a set $P = \{p_1, \ldots, p_{\ell}\}$ of clients and a radius $\mainR' > 0$, 
        in polynomial time,
        we can either compute a pair $(X, p)$ where $X$ is a set of at most $k \cdot \log(\ell)$
        potential centers 
        and $p$ is a client such that $\dist(p_j, X) \leq \mainR'$ for every $j \in [\ell]$ and $\dist(p, X) > (1 + \delta)\mainR'$, or correctly conclude that there is no such pair $(X, p)$ with $|X| \leq k$.
    \end{claim}

    \begin{subproof}
        For every potential center $v$, let $E_v = \{q \in P ~|~ \dist(q, v) \leq \mainR'\}$.
        For every client $p \notin P$, we consider the set system $\mathcal{S}_p = (P, \{E_v ~|~ \dist(v, p) > (1 + \delta)\mainR'\})$.
        We greedily build a set cover $H$ of $\mathcal{S}_p$ by repeatedly adding to $H$ the set from $\mathcal{S}_p$ which contains the maximal number of elements of $P$ which are not yet covered by $H$, until all elements are covered. If $\mathcal{S}_p$ has a set cover of size $h$, the set $H$ returned by the algorithm will have size at most $h \cdot \log(\ell)$ since the number of elements of $P$ not yet covered drops by a factor at least $(1 - 1/h)$ every time we add a set to $H$.

        If for some client $p \notin P$ we find a set cover $H = \{E_{x}, x \in X\}$ of $\mathcal{S}_p$ of size at most $k \cdot \log(\ell)$, we return the pair $(X, p)$. Otherwise, we return that there is no such pair $(X, p)$ with $|X| \leq k$.

        Indeed, if $H = \{E_{x}, x \in X\}$ is a set cover of $\mathcal{S}_p$ then every $x \in X$ satisfies $\dist(x, p) > (1 + \delta) \cdot \mainR'$, so $\dist(p, X) > (1 + \delta) \cdot \mainR'$, and for every $p_j \in P$, there exists $E_x \in H$ such that $p_j \in E_x$, so $\dist(p_j, X) \leq \dist(p_j, x) \leq \mainR'$.
        On the other hand, if there exists such a pair $(X, p)$ with $|X| \leq k$, then $\{E_x, x \in X\}$ forms a set cover of $\mathcal{S}_p$ of size at most $k$ so the greedy algorithm would return a set cover of size at most $k \cdot \log(\ell)$. Thus, if for every $p \in V(G)$ the set cover we find is larger than $k \cdot \log(\ell)$, we correctly report that there is no such pair $(X, p)$ with $|X| \leq k$.
    \end{subproof}

    It suffices to show that given a partial sequence $(X_1, p_1), \ldots, (X_{\ell}, p_{\ell})$ for some radius $\mainR_t$, in polynomial time, we can either extend it with a pair $(X_{\ell+1}, p_{\ell + 1})$ or correctly conclude that there is no such pair with $|X_{\ell + 1}| \leq k$.
    To do so, we simply apply \cref{cl:extend-sequence} to the set $P = \{p_1, \ldots, p_{\ell}\}$ and $\mainR' = \mainR_t$. 
    \cref{cl:length-k-ladder} implies $\ell \leq \left(2c(\lambda \cdot k+2) \cdot (\frac{3}{\varepsilon} + 2)^c\right)^{c+1}$, so $\log(\ell) \leq \log \left[\left(2c(\lambda \cdot k+2) \cdot (\frac{3}{\varepsilon} + 2)^c\right)^{c+1}\right] = \lambda$.
    
    If we get a pair $(X, p)$, we can simply set $(X_{\ell + 1}, p_{\ell + 1}) = (X, p)$, and otherwise we correctly conclude that there is no such pair with $|X| \leq k$.

    This concludes the proof of \cref{thm:coreset-full}.
\end{proof}

\section{Improved bounds for bounded treewidth graphs}\label{sec:wcol-tw}
In this section, we prove \cref{thm:wcol-tw},
following the ideas of~\cite{treewidth,DBLP:journals/corr/abs-2407-12230}.

\thmwcoltw*

\begin{proof}
    Let $(T, \beta)$ be a tree decomposition of $G$ where each bag has size at most $k$.
    Root $T$ in an arbitrary node.

    We define a graph $H$ as follows: $V(H) = \{(v, t) \in V(G) \times V(T)~|~v \in \beta(t)\}$ and $(u, t)$ and $(v, t')$ are adjacent in $H$ if and only if either $t = t'$, or $u=v$ and $tt' \in E(T)$. Every edge $(u, t)(v, t') \in E(H)$ has weight $\dist_G(u, v)$.
    Note that $\dist_H((u,t),(v,t')) = \dist_G(u,v)$.
    We will sometimes refer to vertices in $\beta(t) \times \{t\}$ as vertices in $\beta(t)$.
    
    For a set $Y \subseteq V(H)$, we define $\pi_G(Y) = \{u \in V(G)~|~\exists_{t \in V(T)} (u, t) \in Y\}$.
    
    For every $t \in V(T)$, let $\mathrm{dom}(t)$ be the set of all vertices $(v, t')$ of $H$ where $t'$ is a descendant of $t$ in $T$.
    In the process of the algorithm, 
    we define two families $\mathcal{F}, \mathcal{P}$ of subsets respectively of $V(H)$ and of $V(G)$ and for every $X \in \mathcal{F} \cup \mathcal{P}$ we define a \emph{color} $\mathrm{col}(X) \in [k]$ and a \emph{topmost bag} $\mathrm{top}(X) \in V(T)$.
    We also define colors of vertices of $H$; initially all vertices are uncolored.
    During the course of the algorithm, the set of colored vertices is always equal to $\bigcup \mathcal{F}$.
    
    At every step of the algorithm, for every $t \in V(T)$, by $\mathrm{wdom}(t)$
    we denote the set of all vertices $(v,t') \in \mathrm{dom}(t)$ that do not belong
    to a set $X \in \mathcal{F}$ with $\mathrm{top}(X)$ being an ancestor of $t$. 
    Note that, in particular, $\mathrm{wdom}(t)$ contains all uncolored
    vertices of $\mathrm{dom}(t)$.
    The algorithm is presented as \cref{alg:cover-tw}.
    (This is essentially the same clustering algorithm as in~\cite{DBLP:journals/corr/abs-2407-12230}.)
    
    \begin{algorithm}
    \caption{Algorithm for \cref{thm:wcol-tw}}
    \label{alg:cover-tw}
    \begin{algorithmic}[1]
    \Require An edge-weighted graph $G$, its rooted tree decomposition $(T, \beta)$, a real $\delta > 0$
    
    \State $\mathcal{F} \gets \emptyset$, $\mathcal{P} \gets \emptyset$
    
    \For {$i = 1$ to $k$}
        \While{$\exists t \in V(T)$ s.t. no vertex in $\beta(t)$ has color $i$ and some vertex in $\beta(t)$ is uncolored}
            \State $t \gets$ closest to the root of $T$ such node \label{line:choice-t}
            \State \parbox[t]{\linewidth}{$U \gets \mathrm{wdom}(t)$}
            \State $W \gets$ $U \cap \beta(t)$
            \State $X \gets \mathrm{Ball}_{H[U]}(W, \delta)$
            \State $\mathrm{col}(X) \gets i, \mathrm{top}(X) \gets t$
            \State $\mathcal{F} \gets \mathcal{F} \cup \{X\}$
            \State Color every uncolored vertex in $X$ with color $i$
            \For{$x \in W$}
                \State $Y_x \gets \pi_G(\mathrm{Ball}_{H[U]}(x, \delta))$
                \State $\mathrm{col}(Y_x) \gets i, \mathrm{top}(Y_x) \gets t$
                \State $\mathcal{P} \gets \mathcal{P} \cup \{Y_x\}$
            \EndFor
        \EndWhile
    \EndFor
    
    \State Return $\mathcal{F}, \mathcal{P}$
    
    \end{algorithmic}
    \end{algorithm}

    Whenever we create a set $X$ to be added to $\mathcal{F}$, $X$ consists of uncolored vertices and vertices of sets $X' \in \mathcal{F}$ that were created earlier
    but such that $\mathrm{top}(X')$ is a proper descendant of $t$.
    We also have the following observation.
    \begin{claim}\label{cl:intersections}
        If $X,X' \in \mathcal{F}$ are distinct with $X \cap X' \neq \emptyset$
        and $\mathrm{col}(X) \geq \mathrm{col}(X')$, then 
        \begin{enumerate}
            \item $\mathrm{col}(X') < \mathrm{col}(X)$;
            \item $\mathrm{top}(X)$ is a proper ancestor of $\mathrm{top}(X')$;
            \item $X \cap X' \cap \beta(\mathrm{top}(X')) \neq \emptyset$.
        \end{enumerate}
        In particular, for every $X \in \mathcal{F}$ and every $\mathrm{col}(X) < i \leq k$, 
        there are at most $|\beta(\mathrm{top}(X)) \cap X| \leq k$ elements
        $X' \in \mathcal{F}$ of color $i$ that intersect $X$. Hence, for every $X \in \mathcal{F}$, there are at most
        at most $(k-\mathrm{col}(X))k \leq k^2$ elements $X' \in \mathcal{F}$
        with $\mathrm{col}(X') > \mathrm{col}(X)$ that intersect $X$.
    \end{claim}

    \begin{subproof}
        For the first item, suppose by contradiction that $\mathrm{col}(X) = \mathrm{col}(X')$. 
        Up to renaming $X$ and $X'$, we can assume without loss of generality that $X'$ was created by the algorithm before $X$. 
        Since $\mathrm{col}(X) = \mathrm{col}(X')$, \cref{line:choice-t} implies that $\mathrm{top}(X')$ is closer to the root than $\mathrm{top}(X)$.
        Thus, when the algorithm creates $X$, all vertices in $X'$ are already colored, and $\mathrm{top}(X')$ is not a proper descendant of $\mathrm{top}(X)$, therefore no vertex in $X'$ can be in $X$. 
        This contradicts $X \cap X' \neq \emptyset$.
        
        For the second item, since $\mathrm{col}(X') < \mathrm{col}(X)$ then $X'$ was created before $X$.
        Thus, all vertices in $X'$ were already colored before $X$ was created, so $X \cap X' \neq \emptyset$ implies that $\mathrm{top}(X')$ is a proper descendant of $\mathrm{top}(X)$.
        
        For the third item, since $X \cap X' \neq \emptyset$ there exists $x \in X \cap X'$. 
        Consider the sets $U, W$ which the algorithm was considering when $X$ was created.
        By definition of $X$, there exists $w \in W$ such that $x \in \mathrm{Ball}_{H[U]}(w, \delta)$.
        Let $P$ be a path from $w$ to $x$ of length at most $\delta$ in $H[U]$.
        Let $t \in V(T)$ be the unique node of $T$ such that $x = (\cdot, t)$. 
        Since $x \in X'$ then $t$ is a descendant of $\mathrm{top}(X')$.
                
        Consider a set $Z \in \mathcal{F}$ being as close to the root as possible such that $\mathrm{top}(X)$ is an ancestor of $\mathrm{top}(Z)$, which is itself an ancestor of $\mathrm{top}(X')$, and $\mathrm{col}(Z) \leq \mathrm{col}(X')$ (such a set $Z$ exists since $X'$ is a candidate).
        Therefore, $\mathrm{top}(Z)$ separates $\mathrm{top}(X)$ and $\mathrm{top}(X')$ in $T$ so $\beta(\mathrm{top}(Z))$ separates $w$ and $x$ in $H$, and thus $V(P) \cap \beta(\mathrm{top}(Z)) \neq \emptyset$.
        Let $z$ be the last vertex in $\beta(\mathrm{top}(Z))$ visited by $P$.
        If $z \notin Z$ then $z$ was already colored when $Z$ was created, so there exists a set $Y \in \mathcal{F}$ such that $\mathrm{col}(Y) \leq \mathrm{col}(Z)$ and $z \in Y$.
        Since $z \in Y$ and $z \in \beta(\mathrm{top}(Z))$ then $\mathrm{top}(Y)$ is a proper ancestor of $\mathrm{top}(Z)$.
        By maximality of $\mathrm{top}(Z)$, $\mathrm{top}(Y)$ cannot be a descendant of $\mathrm{top}(X)$, so since $\mathrm{top}(Y)$ is an ancestor of $\mathrm{top}(Z)$ then $\mathrm{top}(Y)$ is a proper ancestor of $\mathrm{top}(X)$.
        Therefore, $z \in Y$ would imply that $z \notin U$, contradicting $V(P) \subseteq U$.
        Thus, $z \in Z$.
        
        Let $P'$ be the restriction of $P$ from $z$ to $x$.
        Since $z$ is the last vertex in $\beta(\mathrm{top}(Z))$ visited by $P$ and since $\mathrm{top}(Z)$ is a descendant of $\mathrm{top}(X)$ and an ancestor of $t$ then $V(P') \subseteq \mathrm{dom}(\mathrm{top}(Z))$.
        Let $U', W'$ be the sets which the algorithm was considering when $Z$ was created.
        Since $z \in Z \cap \beta(\mathrm{top}(Z))$ then $z \in W'$.
        By contradiction, suppose $V(P') \not \subseteq U'$.
        Then, there exists $y \in V(P') \setminus U'$.
        Since $y \in V(P')$ then $y \in \mathrm{dom}(\mathrm{top}(Z))$ so $y \notin U'$ implies that there exists a set $Y \in \mathcal{F}$ which was created before $Z$ such that $y \in Y$ and $\mathrm{top}(Y)$ is a proper ancestor of $\mathrm{top}(Z)$. 
        The maximality of $\mathrm{top}(Z)$ then implies that $\mathrm{top}(Y)$ is a proper ancestor of $\mathrm{top}(X)$, which contradicts that $y \in U$.
        Therefore, $V(P') \subseteq U'$ and thus $V(P') \subseteq \mathrm{Ball}_{H[U']}(W', \delta) = Z$, and in particular $x \in Z$.
        
        If $Z \neq X'$ then $\mathrm{top}(Z)$ is a proper ancestor of $\mathrm{top}(X')$, and $\mathrm{col}(Z) \leq \mathrm{col}(X')$ implies that $Z$ was already created when $X'$ was created, so $x$ was not in $\mathrm{wdom}(\mathrm{top}(X'))$ when $X'$ was created, which implies that $x \notin X'$, a contradiction.
        Therefore, $Z = X'$ and thus $z \in X \cap X' \cap \beta(\mathrm{top}(X'))$.

        The last statement follows from the first item, which implies that any two distinct sets in $\mathcal{F}$ of the same color are disjoint, and from the third item.
    \end{subproof}
    
    We claim that at the end of the algorithm $V(H) = \bigcup \mathcal{F}$, that is,
    all vertices are colored. 
    By contradiction, suppose some $(v, t) \in V(H)$ does not get colored.
    For every $i \in [k]$, when we are in the $i$-th iteration of the For loop, there is always an uncolored vertex in $\beta(t)$ (namely $(v, t)$), so since we exit the While loop there must be some vertex $(v_i, t)$ of color $i$. Every vertex receives at most one color so $(v_1, t), \ldots, (v_k, t), (v, t)$ are pairwise distinct and belong to $\beta(t)$, contradicting that $|\beta(t)| \leq k$. Thus, $V(H) = \bigcup\mathcal{F}$.

    Every set $X \in \mathcal{F}$ is defined as $\mathrm{Ball}_{H[U]}(U \cap \beta(\topm(X)), \delta)$ for some set $U$. Thus, every connected component of $H[X]$ contains a vertex of $\beta(\topm(X))$ and for every $x \in X$, $\dist_{H[X]}(x, \beta(\topm(X))) \leq \delta$. Furthermore, $X \subseteq \dom(\topm(X))$.

    We say that $X \in \mathcal{F}$ \emph{generates} $Y \in \mathcal{P}$ if $Y$ was added to $\mathcal{P}$ in the iteration of the While loop where $X$ was added to $\mathcal{F}$. Note that if $X$ generates $Y$ then $\mathrm{top}(X) = \mathrm{top}(Y)$, and $Y \subseteq \pi_G(X)$.
    We observe that for every $X \in \mathcal{F}$ and $(v,t) \in X$, 
    there exists $Y \in \mathcal{P}$ generated by $X$ with $v \in Y$. Indeed,
    consider the point during \cref{alg:cover-tw} when $X$ was defined and added to $\mathcal{F}$. 
    Let $U, W$ be the sets stored by the algorithm at that point. 
    Then, $X = \mathrm{Ball}_{H[U]}(W, \delta)$ so there exists $x \in W$ such that $(v, t) \in \mathrm{Ball}_{H[U]}(x, \delta)$ so $v \in \pi_G(\mathrm{Ball}_{H[U]}(x, \delta)) = Y_x \in \mathcal{P}$. 
    In particular, we have $V(G) = \bigcup_{Y \in \mathcal{P}} Y$.
    
    We say that a node $s \in V(T)$ is a \emph{topmost node of color $i$} if there exists $X \in \mathcal{F}$ such that $\mathrm{top}(X) = s$ and $\mathrm{col}(X) = i$. If there exists such an $X$, it is unique (as in a single iteration of the While loop, the algorithm puts into $X$ all uncolored vertices of $\beta(s)$) and we denote it by $X(s)$.

    \begin{claim} \label{cl:same-color-far-away}
        Let $X_1, X_2 \in \mathcal{F}$ be such that $\mathrm{col}(X_1) = i \geq \mathrm{col}(X_2)$.
        Let $t_1 = \mathrm{top}(X_1)$ and $t_2 = \mathrm{top}(X_2)$ and suppose $t_1$ is a proper ancestor of $t_2$.
        Let $P$ be a path in $H[\mathrm{dom}(t_1)]$ from $X_1 \cap \beta(t_1)$ to $X_2 \cap \beta(t_2)$ that does not visit any $X \in \mathcal{F}$ such that $\mathrm{top}(X)$ is a proper ancestor of $t_1$ and $\mathrm{col}(X) < i$.
        Then, either $P$ has length more than $\delta$ or $X_1 \cap X_2 \neq \emptyset$.
    \end{claim}

    \begin{subproof}
        Without loss of generality, we can assume that $P$ is a shortest such path, and among them one with as few edges as possible.
        Let $Q = (t_1 = s_0, s_1, \ldots, s_k = t_2)$ be the path from $t_1$ to $t_2$ in $T$.
    
        First, since every edge of $H$ is a shortest path between its endpoints and since every $H[\beta(s)]$ is a clique, $P$ having as few edges as possible implies that $P$ visits at most 2 vertices in every $\beta(s)$.
        For the same reason, if $P$ visits two vertices in $\beta(s)$ then these two vertices are consecutive in $P$.
        This implies that if $P$ visits $\beta(s)$ then $s \in V(Q)$, and $P$ visits $\beta(s_0)$ then $\beta(s_1)$ and so on until $\beta(s_k)$.
    
        Consider the iteration of \cref{alg:cover-tw} where $X_1$ was defined and added to $\mathcal{F}$, and let $U, W$ be the sets the algorithm was considering at that moment. The node $t$ under consideration is $t_1$.
        By definition, $W = U \cap \beta(t_1)$ and $X_1 \subseteq U$.
        Observe that $W = \beta(t_1) \cap X_1$, i.e., $X_1$ gets all vertices
        of $\beta(t_1)$ that are uncolored prior to this iteration.
        
        In particular, from the restriction on the vertices of $P$ we infer that $P$ needs to visit a vertex of $W$ before it can visit any vertex outside $\mathrm{dom}(t_1)$.
        Hence, $V(P) \subseteq \mathrm{dom}(t_1)$, and, again from the restriction on the vertices of $P$, we have $V(P) \subseteq U$. 
        If $P$ has length at most $\delta$ then $V(P) \subseteq \mathrm{Ball}_{H[U]}(W, \delta) = X_1$, which would imply $X_1 \cap  X_2 \neq \emptyset$, as desired.
    \end{subproof}
    
    Fix an ordering $\trianglelefteq$ on $\mathcal{F}$ such that if $\mathrm{top}(X_1)$ is a proper ancestor of $\mathrm{top}(X_2)$ then $X_1 \triangleleft X_2$.
    Fix an ordering $\preceq$ on $\mathcal{P}$ such that if $X_1 \in \mathcal{F}$ generates $Y_1 \in \mathcal{P}$, $X_2 \in \mathcal{F}$ generates $Y_2 \in \mathcal{P}$ and $X_1 \triangleleft X_2$ then $Y_1 \prec Y_2$. This way, if $\mathrm{top}(Y_1)$ is a proper ancestor of $\mathrm{top}(Y_2)$ then $Y_1 \prec Y_2$. 
    For $Y \in \mathcal{P}$, let $Y_\ast := Y \setminus \bigcup_{Y' \prec Y} Y'$,
    let $\mathcal{P}_\ast = \{Y_\ast, Y \in \mathcal{P}\}$.
    The ordering $\preceq$ naturally projects to an ordering of $\mathcal{P}_\ast$,
    and somehow abusing the notation we also denote this projection by $\preceq$.

    Since every $Y \in \mathcal{P}$ is defined as $\pi_G(\mathrm{Ball}_{H[U]}(x,\delta))$, 
    every such $Y$ has (strong) diameter at most $2\delta$
    and hence every $Y_\ast \in \mathcal{P}_\ast$ has weak diameter at most $2\delta$.

    Setting $\delta = \tinyR / 2$, it suffices to prove that for every $Y_\ast \in \mathcal{P}_\ast$ and every $\mainR > \tinyR$ it holds that
    \begin{equation}\label{eq:wcol-tw}
    \wcol_{\mainR}(G, \mathcal{P}_\ast, \preceq, Y_\ast) \leq 
    \min \left( k2^k\binom{k+\lceil \mainR / \delta \rceil}{k}, 
    k \cdot \left(2 \lceil \mainR/\delta \rceil + k + 1 \right)^{3 \lceil \mainR / \delta \rceil + 3}
    \right).
    \end{equation}

    Fix $\mainR > \tinyR$, $Y_\ast \in \mathcal{P}_\ast$, let $Y$ be its corresponding part in $\mathcal{P}$, $t =\mathrm{top}(Y)$ and $X = X(t)$.
    \begin{claim}
        Let $Y_\ast' \in \WReach{\mainR}{G}{\mathcal{P}_\ast}{\preceq}{Y_\ast}$,
        $Y' \in \mathcal{P}$ be its corresponding part,
        and $t' = \mathrm{top}(Y')$.
        Then, $t'$ is a topmost node which is an ancestor of $t$ and in $H - \left(\bigcup_{Z \triangleleft X(t')} Z\right)$, the distance between $X(t') \cap \beta(t')$
        and $X$ is at most $\mainR + \delta$.
    \end{claim}

    \begin{subproof}
        Let $P$ be a path of length at most $\mainR$ from $Y_\ast$ to $Y_\ast'$ in
        $G-\left(\bigcup_{Z_\ast \prec Y_\ast'} Z_\ast\right)$.
        Then, $P$ is a path of length at most $\mainR$ from $Y$ to $Y'$ in
        $G-\left(\bigcup_{Z \prec Y'} Z\right)$.
        Let $X' \in \mathcal{F}$ generate $Y'$. Then, 
        $\mathrm{top}(X') = \mathrm{top}(Y') = t'$ so $t'$ is a topmost node and $X' = X(t')$.
        
        We first argue that $t$ and $t'$ are in an ancestor-descendant relation in $T$.
        By contradiction, suppose not and let $t_0$ be their lowest common ancestor in $T$.
        Let $\alpha$ be the endpoint of $P$ in $Y$. Since $\alpha \in Y \subseteq \pi_G(X) \subseteq \pi_G(\mathrm{dom}(t))$, there exists a node $t_{\alpha}$ of $T$ such that $(\alpha, t_{\alpha}) \in \mathrm{dom}(t)$, so $\alpha \in \beta(t_{\alpha})$ and $t_{\alpha}$ is a descendant of $t$.
        Similarly, if $\gamma$ is the endpoint of $P$ in $Y'$, there exists a descendant $t_{\gamma}$ of $t'$ such that $\gamma \in \beta(t_{\gamma})$.
        Then, $t_0$ is a node of the unique path from $t_\alpha$ to $t_\gamma$ in $T$ so there exists $u \in V(P) \cap \beta(t_0)$. 
        Since $(u, t_0) \in V(H)$, there exists $F \in \mathcal{F}$ such that $(u, t_0) \in F$. Therefore, $(u, t_0) \in \mathrm{dom}(\mathrm{top}(F))$ so $\mathrm{top}(F)$ is an ancestor of $t_0$.
        There exists $Z \in \mathcal{P}$ generated by $F$ such that $u \in Z$.
        Since $u \in V(P)$ then $Z \succeq Y'$. 
        However, $\mathrm{top}(Z) = \mathrm{top}(F)$ is an ancestor of $t_0$, so a proper ancestor of $t'$. Thus, $\mathrm{top}(Z)$ is a proper ancestor of $\mathrm{top}(Y')$ in $T$, so $Z \prec Y'$, a contradiction.
        Thus, $t$ and $t'$ are comparable in $T$, and $Y' \preceq Y$ implies that $t'$ is an ancestor of $t$.
        
        Finally, $P$ witnesses that in $H - \left(\bigcup_{Z \triangleleft X'} Z\right)$
        the distance from $X$ to $X'$ is at most $\mainR$, 
        thus the distance from $X$ to $X' \cap \beta(t')$ is at most $\mainR+\delta$.
    \end{subproof}

    \begin{claim}
        The number of topmost nodes $t'$ which are ancestors of $t$ such that in $H - \left(\bigcup_{Z \triangleleft X(t')} Z\right)$
        the distance between $X(t') \cap \beta(t')$ and $X$ is at most $\mainR + \delta$ is bounded by
        \begin{enumerate}
            \item $2^k\binom{k+\lceil\mainR/\delta\rceil}{k} \leq (2\lceil \mainR/\delta \rceil + 2k)^k$, and
            \item $(2 \lceil \mainR/\delta \rceil + k + 1)^{3\lceil \mainR/\delta \rceil + 3}$.
        \end{enumerate}
    \end{claim}

    \begin{subproof}
        Consider such a $t'$. 
        Let $Q$ be the path from $t$ to $t'$ in $T$.
        For $i \in [k]$, let $t_i$ be the topmost node of color $\leq i$ in $V(Q)$ which is the closest to the root; set $t_i = t$ if there is no such node. As $t'$ is a topmost node, $t_k = t'$.
        Set $t_0 = t$ and note that $t=t_0,t_1,\ldots,t_k=t'$ appear in this order
        on $Q$ (some of these nodes may be equal). 
        
        For $i \in [k]$, let $S_i$ be the set of topmost nodes of color $i$ between $t_{i-1}$ and $t_i$ (inclusive) in $Q$ and let $a_i = |S_i|$.
        Let $S = S_1 \cup \ldots \cup S_k$.
        %
        %Denote $S_i = \{t_i^{a_i}, t_i^{a_i-1}, \ldots, t_i^1\}$, where $t_i^1 = t_i$ and $t_i^a$ is an ancestor
        %of $t_i^b$ for every $1 \leq a \leq b \leq a_i$. 

        Let $P$ be a shortest path from $X$ to $X(t') \cap \beta(t')$ in $H - \left(\bigcup_{Z \triangleleft X(t')} Z\right)$. $P$ has length at most $\mainR + \delta$.
        
        We claim that for every $s \in S$, we have $\beta(s) \cap V(P) \subseteq X(s)$.
        Let $s \in S$, $i = \mathrm{col}(X(s))$ and consider $x \in \beta(s) \cap V(P)$.
        Consider $Z \in \mathcal{F}$ of minimum color such that $x \in Z$, and let $t_Z = \mathrm{top}(Z)$. Since $x \in V(P)$ then $Z \trianglerighteq X(t')$. 
        Note that $t_Z$ is an ancestor of $s$ so $t_Z$ and $t'$ are comparable and therefore $t'$ is an ancestor of $t_Z$, which implies $t_Z \in V(Q)$.
        Since $x \in \beta(s)$ and $s$ is a topmost node of color $i$ then $\mathrm{col}(x) \leq i$ so $\mathrm{col}(Z) \leq i$.
        Thus, $t_Z$ is a topmost node of color $\leq i$. Since $t_Z$ is an ancestor of $s$, it is a proper ancestor of $t_{i-1}$ so $t_Z$ has color $\geq i$, thus $t_Z$ has color $i$, so $\mathrm{col}(Z) = i$. 
        If $Z \neq X(s)$ then $Z$ was created before $X(s)$ so $x \in \beta(s)$ was colored with color $i$ before $X(s)$ (of color $i$) was created, which is impossible.
        Therefore, $Z = X(s)$ and $x \in X(s)$.

        For $i \in [k]$, denote by $s^i_1, \ldots, s^i_{a_i}$ the nodes in $S_i$ in the order in which they appear in $Q$ (from $t$ to $t'$).
        Then, $Q$ visits the $s^i_j$ by lexicographic order of $(i, j)$.
        Let $x$ be the endpoint of $P$ in $X$. 
        Then, $x \in X \subseteq \dom(t)$ so there exists a descendant $t''$ of $t$ such that $x \in \beta(t'')$.
        Since $P$ is a path in $H$ between a vertex in $\beta(t'')$ and a vertex in $\beta(t')$ then $P$ visits every $\beta(q)$ for $q \in V(Q)$, hence every $\beta(s)$ for $s \in S$. If $s^i_j \in S$, denote by $v^i_j$ the first vertex in $V(P) \cap \beta(s^i_j) \subseteq X(s^i_j)$.

        By construction of $H$, $P$ visits the $v^i_j$ by lexicographically increasing $(i, j)$.
        Let $v^i_j$ and $v^{i'}_{j'}$ be two vertices visited by $P$ such that $(i,j)$ is lexicographically earlier
        than $(i',j')$. We claim that if $X(s^i_j) \cap X(s^{i'}_{j'}) = \emptyset$, then the subpath of $P$
        between $v^i_j$ and $v^{i'}_{j'}$ is longer than $\delta$. 
        
        %For $i \in [k]$ and $j < a_i$, let $P^i_j$ be the restriction of $P$ between $v^i_j$ and $v^i_{j+1}$.
        %Thus, $P$ has length at least $\sum_{i=1}^k \sum_{j=1}^{a_i-1} \mathrm{length}(P^i_{j})$.

        We have $X(s^i_j), X(s^{i'}_{j'}) \in \mathcal{F}$, $\mathrm{col}(X(s^i_j)) = i$, $\mathrm{col}(X(s^{i'}_{j'})) = i'$, $i \leq i'$, $s^i_j = \mathrm{top}(X(s^i_{j}))$, $s^{i'}_{j'} = \mathrm{top}(X(s^{i'}_{j'}))$ and $s^{i'}_{j'}$ is a proper ancestor of $s^i_j$. 
        Let $P'$ be the subpath of $P$ between $v^i_j$ and $v^{i'}_{j'}$.
        By minimality of $s^{i'}_{j'}$, $P'$ is a path in $H[\mathrm{dom}(s^{i'}_{j'})]$ between $X(s^{i'}_{j'}) \cap \beta(s^{i'}_{j'})$ and $X(s^i_{j})\cap \beta(s^{i}_{j})$.
        Furthermore, if $P'$ visits some $Z \in \mathcal{F}$ such that $\mathrm{top}(Z)$ is a proper ancestor of $s^{i'}_{j'}$ and $\mathrm{col}(Z) < i'$ then $\mathrm{top}(Z)$ would be a proper ancestor of $t_{i'-1}$ and therefore by definition of $t_{i'-1}$, $\mathrm{top}(Z)$ would be a proper ancestor of $t'$, which in turn would imply $Z \triangleleft X(t')$, contradicting that $P$ is a path in $H - \left(\bigcup_{Z \triangleleft X(t')} Z\right)$.
        Therefore, $P'$ does not visit any $Z \in \mathcal{F}$ such that $\mathrm{top}(Z)$ is a proper ancestor of $s^{i'}_{j'}$ and $\mathrm{col}(Z) < i'$.
        Therefore, by \cref{cl:same-color-far-away}, $P'$ has length $> \delta$.
        
        Since for every $i \in [k]$ and $j < a_i$, the sets $X(s^i_j)$ and $X(s^i_{j+1})$ are disjoint (because they have the same color), we have
        \[ \mainR + \delta \geq \mathrm{length}(P) \geq \sum_{i, a_i \geq 1} (a_i - 1) \delta. \]
        Furthermore, one of these inequalities is strict, hence 
        \[ \sum_{i, a_i \geq 1} (a_i - 1) \leq \lceil \mainR/\delta \rceil. \]
        The number of sequences $(a_1, \ldots, a_k)$ of non-negative integers satisfying this condition is \[ \sum_{j=0}^k \binom{k}{j} \binom{\lceil \mainR/\delta \rceil+j}{j} \leq 2^k \cdot \binom{k+\lceil \mainR/\delta \rceil}{k} \].

        Since the sequence $(a_1, \ldots, a_k)$ uniquely identifies $t'$, this proves the first promised bound.
        For the second bound, we need to refine the sequence of the nodes $s_j^i$ further.

        For a node $s_j^i$, let $f(i,j) = (i',j')$ be the lexicographically minimum tuple strictly greater than $(i,j)$
        such that $X(s_j^i) \cap X(s_{j'}^{i'}) = \emptyset$, and $f(i,j) = \bot$ if no such tuple $(i',j')$ exists. 
        Furthermore, let $g(i,j)$ be the direct predecessor of $f(i,j)$ in the lexicographic order (and $g(i,j) = \bot$ if $f(i,j) = \bot$).
        Observe that if $j < a_i$, then $f(i,j) = (i,j+1)$ and $g(i,j) = (i,j)$, but $f(i,a_i)$ may be different than $(i+1,1)$.  

        Define a sequence $\alpha_0,\alpha_1,\ldots,\alpha_\ell$ with
        $\alpha_0 = (0,1)$, $\alpha_{a+1} = f(\alpha_a)$, $f(\alpha_{\ell - 1}) \neq \bot$ and $f(\alpha_\ell) = \bot$. 
        The definition of $f$ and \cref{cl:same-color-far-away} together imply that
        if $f(i,j) = (i',j')$, then the subpath of $P$ between $v_j^i$ and $v_{j'}^{i'}$ has length
        more than $\delta$ (with the same argument as previously).
        It follows that $\ell \leq \lceil \mainR/\delta \rceil$. 
        Furthermore, for $0 \leq a < \ell$, let $\beta_{a+1} = g(\alpha_a)$. 
        
        For brevity, by $s(\alpha_a)$, $X(\alpha_a)$, and $\mathrm{col}(\alpha_a)$ we denote $s_j^i$, $X(s_j^i)$, and $\mathrm{col}(X(s_j^i))$ where $\alpha_a = (i,j)$ and similarly
        we define $s(\beta_a)$, $X(\beta_a)$, and $\mathrm{col}(\beta_a)$. Note that 
        \[ \mathrm{col}(\alpha_0) \leq \mathrm{col}(\beta_1) \leq \mathrm{col}(\alpha_1) \leq \ldots \leq \mathrm{col}(\beta_\ell) \leq \mathrm{col}(\alpha_\ell). \]
        As $\ell \leq \lceil \mainR/\delta \rceil$, the number of choices of $\ell$ and of
        the sequence $\mathrm{col}(\alpha_0), \mathrm{col}(\beta_1), \mathrm{col}(\alpha_1), \ldots, \mathrm{col}(\beta_\ell), 
        \mathrm{col}(\alpha_{\ell})$ is bounded by 
        \[ \sum_{\ell = 0}^{\lceil \mainR/\delta \rceil} \binom{2\ell+1 +k}{k}. \]
        Observe that for $0 \leq a < \ell$ and $\alpha_a$, it holds that
        $X(\alpha_a) \cap X(g(\alpha_a)) \neq \emptyset$;
        by \cref{cl:intersections}, for every fixed value of $\mathrm{col}(\beta_{a+1})$, there are at most
        $k$ options for a set of $\mathcal{F}$ of this color that intersects $X(\alpha_a)$. 
        Given $\beta_{a+1}$, the value $\alpha_{a+1}$ is defined
        uniquely: $\alpha_{a+1}$ is the successor of $\beta_{a+1}$ in the lexicographic order.

        Therefore, the number of choices for $\ell$ and the sequence $s(\alpha_0), s(\beta_1), s(\alpha_1), \ldots, s(\beta_\ell), s(\alpha_\ell)$ is bounded by
        \[ \sum_{\ell = 0}^{\lceil \mainR/\delta \rceil} \binom{2\ell+1 +k}{k} \cdot k^\ell \leq 
        \binom{2 \lceil \mainR/\delta \rceil +k + 1}{k} \cdot (k+1)^{\lceil \mainR/\delta \rceil} \leq (2 \lceil \mainR/\delta \rceil + k + 1)^{3\lceil \mainR/\delta \rceil + 1}. \]
        Finally, given $s(\alpha_\ell)$, there are at most $k^2$ choices for $t'$ as the topmost node
        of color at least $\mathrm{col}(\alpha_\ell)$ that is an ancestor of $s(\alpha_\ell)$ such that $X(t') \cap X(\alpha_\ell) \neq \emptyset$. This gives the second promised bound.
    \end{subproof}

    Putting together the previous two claims and noting that for every $t' \in V(T)$ 
    there are at most $k$ sets $Y'_\ast \in \mathcal{P}_\ast$ with $\mathrm{top}(Y'_\ast) = t'$
    yields the desired bound.
\end{proof}

\bibliographystyle{plain}
\bibliography{biblio.bib}

\end{document}